%% file: JLPrivacy.tex
\documentclass[11pt]{article}

\usepackage{paper}
\usepackage{bbding}
\usepackage[boxruled,linesnumbered]{algorithm2e}
\usepackage{color}

\newcommand{\T}{\ensuremath{^{\mathsf{\scriptscriptstyle T}}}}
\newcommand{\mpinv}{\ensuremath{^{\dag}}}
\newcommand{\PDF}{\ensuremath{\mathsf{PDF}}}
\newcommand{\1}{\ensuremath{\mathbf{1}}}

\renewcommand{\paragraph}[1]{\vspace{1mm}\noindent\textbf{#1}}

\ifx\ConferenceVersion\undefined
\usepackage{fullpage}
\else 
\fi 

\begin{document}

\begin{titlepage}
\title{The Johnson-Lindenstrauss Transform Itself Preserves Differential Privacy}

\author{
Jeremiah Blocki\thanks{Supported in part by the National Science Foundation Science and Technology Center TRUST as well as a NSF Graduate fellowship.}\hspace{1cm} Avrim Blum\thanks{Supported in part by the National Science Foundation under grants
CCF-1101215 and CCF-1116892.} \hspace{1cm} Anupam Datta\thanks{Supported in part by the National Science Foundation Science and Technology Center TRUST.} \hspace{1cm} Or Sheffet\thanks{Supported in part by the National Science Foundation under grants
CCF-1101215 and CCF-1116892 as well as by the MSR-CMU Center for
Computational Thinking.} \\
Carnegie Mellon University\\
\texttt{\small \{jblocki@cs, avrim@cs, danupam@andrew, osheffet@cs\}.cmu.edu}
}
\date{\today}
\maketitle

\begin{abstract}
This paper proves that an ``old dog'', namely -- the classical Johnson-Lindenstrauss transform, ``performs new tricks'' -- it  gives a novel way of preserving differential privacy. We show that if we take two databases, $D$ and $D'$, such that (i) $D'-D$ is a rank-$1$ matrix of bounded norm and (ii) all singular values of $D$ and $D'$ are sufficiently large, then multiplying either $D$ or $D'$ with a vector of iid normal Gaussians yields two statistically close distributions in the sense of differential privacy. Furthermore, a small, deterministic and  \emph{public} alteration of the input is enough to assert that all singular values of $D$ are large.

We apply the Johnson-Lindenstrauss transform to the task of approximating cut-queries: the number of edges crossing a $(S,\bar S)$-cut in a graph. We show that the JL transform allows us to \emph{publish a sanitized graph} that preserves edge differential privacy (where two graphs are neighbors if they differ on a single edge) while adding only $O(|S|/\epsilon)$ random noise to any given query (w.h.p). Comparing the additive noise of our algorithm to existing algorithms for answering cut-queries in a differentially private manner, we outperform all others on small cuts ($|S| = o(n)$).

We also apply our technique to the task of estimating the variance of a given matrix in any given direction. The JL transform allows us to \emph{publish a sanitized covariance matrix} that preserves differential privacy w.r.t bounded changes (each row in the matrix can change by at most a norm-$1$ vector) while adding random noise of magnitude independent of the size of the matrix (w.h.p). In contrast, existing algorithms introduce an error which depends on the matrix dimensions.
\end{abstract}
\thispagestyle{empty}
\end{titlepage}

\section{Introduction}
\label{sec:intro}

The celebrated Johnson Lindenstrauss transform~\cite{JL84} is widely used across many areas of Computer Science. A very non-exhaustive list of related applications include metric and graph embeddings~\cite{Bourgain_1985,Linial:1994}, computational speedups~\cite{Sarlos06, vempala2005random}, machine learning~\cite{BalcanBV06,Schulman:2000}, information retrieval~\cite{Papadimitriou98latentsemantic}, nearest-neighbor search~\cite{Kleinberg97twoalgorithms, Indyk98approximatenearest, Ailon:2006:ANN:1132516.1132597}, and compressed sensing~\cite{Baraniuk_Davenport_DeVore_Wakin_2008}. This paper unveils a new application of the Johnson Lindenstrauss transform -- it also preserves differential privacy.

Consider a scenario in which a trusted curator gathers personal information from $n$ individuals, and wishes to release statistics about these individuals to the public without compromising any individual's privacy. \emph{Differential privacy}~\cite{Dwork06calibratingnoise} provides a robust guarantee of privacy for such data releases. It guarantees that for any two neighboring databases (databases that differ on the details of any single individual), the curator's distributions over potential outputs are statistically close (see formal definition in Section~\ref{sec:preliminaries}). By itself, preserving differential privacy isn't hard, since the curator's answers to users' queries can be so noisy that they obliterate any useful data stored in the database. Therefore, the key research question in this field is to provide tight \emph{utility and privacy tradeoffs}. 

The most basic technique that preserves differential privacy and gives good utility guarantees is to add relatively small Laplace or Gaussian noise to a query's true answer. 
This simple technique lies at the core of an overwhelming majority of algorithms that preserve differential privacy. In fact, many differentially private algorithms follow a common outline. They take an existing algorithm and revise it by adding such random noise each time the algorithm operates on the sensitive data. Proving that the revised algorithm preserves differential privacy is almost immediate, because differential privacy is composable. On the other hand, providing good bounds on the revised algorithm's utility follows from bounding the overall noise added to the algorithm, which is often difficult. This work takes the complementary approach. We show that an existing algorithm preserves differential privacy provided we slightly alter the input in a reversible way. Our analysis of the algorithm's utility is immediate, whereas privacy guarantees require a non-trivial proof. 

We prove that by multiplying a given database with a vector of iid normal Gaussians, we can output the result while preserving differential privacy (assuming the database has certain properties, see ``our technique''). This technique is no other than the Johnson-Lindenstrauss transform, and it's guaranteed to preserve w.h.p the $L_2$ norm of the given database up to a small multiplicative factor. Therefore, whenever answers to users' queries can be formalized as the length of the product between the given database and a query-vector, utility bounds are straight-forward. 

For example, consider the case where our input is composed of $n$ points in $\mathbb{R}^d$ given as a $n\times d$ matrix. We define two matrices as neighbors if they differ on a single row and the norm of the difference is at most $1$.\footnote{This notion of neighboring inputs, also considered in~\cite{McSherryM09, HR11}, is somewhat different than the typical notion of privacy, allowing any individual to change her attributes arbitrarily.} Under this notion of neighbors, a simple privacy preserving mechanism allows us to output the mean of the rows in $A$, but what about the covariance matrix $A\T A$? We prove that the JL transform gives a $(\epsilon,\delta)$-differentially private algorithm that outputs a sanitized covariance matrix. Furthermore, for \emph{directional variance queries}, where users give a unit-length vector $x$ and wish to know the variance of $A$ along $x$ (see definition in Section~\ref{sec:preliminaries}), we give utility bounds that are \emph{independent of $d$ and $n$}. In contrast, all other differentially private algorithms that answer directional variance queries have utility guarantees that depend on $d$ or $n$. Observe that our utility guarantees are somewhat weaker than usual. Recall that the JL lemma guarantees that w.h.p lengths are preserved up to a small multiplicative error, so for each query our algorithm's estimation has w.h.p small multiplicative error and additional additive error.

A special case of directional variance queries is \emph{cut-queries} of a  graph. Suppose our database is a graph $G$ and users wish to know how many edges cross a $(S,\bar S)$-cut. Such a query can be formalized by the length of the product $E_G \1_S$, where $E_G$ is the \emph{edge-matrix} of $G$ and $\1_S$ is the indicator vector of $S$ (see Section~\ref{sec:preliminaries}). We prove that the JL transform allows us to publish a perturbed Laplacian of $G$ while preserving $(\epsilon,\delta)$-differential privacy, w.r.t two graphs being neighbors if they differ only on a single edge. Comparing our algorithm to existing algorithms, we show that we add (w.h.p) $O(|S|)$ random noise to the true answer (alternatively: w.h.p we add only constant noise to the query $\frac {\1_S\T E_G\T E_G \1_S}{\1_S\T \1_S}$). In contrast, all other algorithms add noise proportional to the number of vertices (or edges) in the graph.

\paragraph{Our technique.} It is best to demonstrate our technique on a toy example. Assume $D$ is a database represented as a $\bitsn$-vector, and suppose we sample a vector $Y$ of $n$ iid normal Gaussians and publish $X = Y\T D$. Our output is therefore distributed like a Gaussian random variable of $0$ mean and variance $\sigma^2 = \|D\|^2$. Assume a single entry in $D$ changes from $0$ to $1$ and denote the new database as $D'$. Then $X' = Y\T D'$ is distributed like a Gaussian of $0$-mean and variance $\lambda^2 = \|D\|^2+1$. Comparing $\PDF_X(x) = (2\pi\sigma^2)^{-1/2} \exp(-x^2/(2\sigma^2))$ to $\PDF_{X'}(x) = (2\pi\lambda^2)^{-1/2} \exp(-x^2/(2\lambda^2))$ we have that $\forall x, ~ \sqrt{\lambda^2/\sigma^2} \PDF_{X'}(x) \geq \PDF_X(x) \geq \exp(-\frac {x^2}{2\sigma^2} \cdot \frac 1{ \lambda^2}) \PDF_{X'}(x)$. Using concentration bounds on Gaussians we deduce that if $\lambda^2 > \sigma^2  = \Omega(\log(1/\delta)/\epsilon)$, then w.p $\geq 1-\delta$ both $\PDF$s are within multiplicative factor of $e^{\pm\epsilon}$. We now repeat this process $r$ times (setting $\epsilon,\delta$ accordingly) s.t. the JL lemma assures that (after scaling) w.h.p we output a vector of norm $(1\pm\eta) \|D\|^2$ for a given $\eta$. We get utility guarantees for publishing the number of ones in $D$ while preserving $(\epsilon,\delta)$-differential privacy.

Keeping with our toy example, one step remains -- to convert the above analysis so that it will hold for any database, and not only databases with $w \stackrel{\mathrm{def}} = \log(1/\delta)/\epsilon$ many ones. One way is to append the data with $w$ one entries, but observe: this ends up in outputting $X + N$ where $N$ is random Gaussian noise! In other word, appending the data with ones makes the above technique worse (noisier) than the classical technique of adding random Gaussian noise. Instead, what we do is to ``translate the database''. We apply a simple \emph{deterministic} affine transformation s.t. $D$ turns into a $\{\sqrt {\frac w n}, 1\}^n$-vector. Applying the JL algorithm to the translated database, we output a vector whose norm squared is $\approx (1\pm \eta)(\|D\|^2+w)$. Clearly, users can subtract $w$ from the result, and we end up with $\eta w$ additive random noise (in addition to the multiplicative noise).\footnote{Observe that in this toy example, our $O(\log(1/\delta)/\epsilon)$ noise bound is still worse than the noise bound of $O(\sqrt{\log(1/\delta)}/\epsilon)$ one gets from adding Gaussian noise. However, in the applications detailed in Sections~\ref{sec:perturbed_laplacian} and~\ref{sec:covariance_matrix}, the idea of changing the input will be the key ingredient in getting noise bounds that are independent of $n$ and $d$.}

It is tempting to think the above analysis suffices to show that privacy is also preserved in the multidimensional case. After all, if we  multiply the edge matrix of a graph $G$ with a vector of iid normal Gaussians, we get a vector with each entry distributed like a Gaussian; and if we replace $G$ with a neighboring $G'$, we affect only two entries in this vector. Presumably, applying the previous analysis to both entries suffices to prove we preserve differential privacy. But this intuition is false. Multiplying $E_G$ with a random vector does not result in $n$ independent Gaussians, but rather in one multivariate Gaussian. This is best illustrated with an example. Suppose $G$ is a graph and $S$ is a subset of nodes s.t. no edge crosses the $(S,\bar S)$-cut. Therefore $E_G\1_S$ is the zero-vector, and no matter what random projection we pick, $Y\T E_G \1_S = 0$. In contrast, by adding a single edge that crosses the $(S,\bar S)$-cut, we get a graph $G'$ s.t. $\Pr[Y\T E_{G'}1_S \neq 0] = 1$.

\paragraph{Organization.} Next we detail related work. Section~\ref{sec:preliminaries} details important notations and important preliminaries. In Sections~\ref{sec:perturbed_laplacian} and~\ref{sec:covariance_matrix} we convert the above univariate intuition to the multivariate Gaussian case. 
\ifx\ConferenceVersion\undefined
Section~\ref{sec:perturbed_laplacian} describes our results for graphs and cut-queries, and in Section~\ref{subsec:comparison_of_algorithsm} we compare our method to other algorithms. Section~\ref{sec:covariance_matrix} details the result for directional queries (the general case), then a comparison with other algorithms.
\else 
Section~\ref{sec:perturbed_laplacian} describes our results for graphs and cut-queries, and Section~\ref{sec:covariance_matrix} details the result for directional queries (the general case). Due to space limitations, the comparison of our algorithms with existing algorithms is deferred to the appendix.
\fi 
Even though there are clear similarities between the  analyses in Sections~\ref{sec:perturbed_laplacian} and~\ref{sec:covariance_matrix}, we provide both because the graph case is simpler and analogous to the univariate Gaussian case. Suppose $G$ and $G'$ are two graphs without and with a certain edge resp., then $G$ induces the multivariate Gaussian with the ``smaller'' variance, and $G'$ induces the multivariate Gaussian with the ``larger'' variance. In contrast, in the general case there's no notion of ``smaller'' and ``larger'' variances. Also, the noise bound in the general case is larger than the one for the graph case, and the theorems our analysis relies on are more esoteric. 
Section~\ref{sec:open_problems} concludes with a discussion and open problems.

\subsection{Related Work}
\label{subsec:related}

\input{JLPrivacy_related}


\section{Basic Definitions, Preliminaries and Notations}
\label{sec:preliminaries}

\paragraph{Privacy and utility.} In this work, we deal with two types of inputs: $[0,1]$-weighted graphs over $n$ nodes and $n\times d$ real matrices. (We treat $w_{a,b} = 0$ as no edge between $a$ and $b$). Trivially extending the definition in~\cite{nissim2007smooth, KRSY11}, two weighted $n$-nodes graphs $G$ and $G'$ are called \emph{neighbors} if they differ on the weight of a single edge $(a,b)$. Like in~\cite{HR11}, two $n\times d$-matrices are called \emph{neighbors} if all the coordinates on which $A$ and $A'$ differ lie on a single row $i$, s.t. $\|A_{(i)}-A'_{(i)}\|^2 \leq 1$, where $A_{(i)}$ denotes the $i$-th row of $A$. 

\begin{definition}
\label{def:privacy}
An algorithm $\textsf{ALG}$ which maps inputs into some range $\mathcal{R}$ maintains \emph{$(\epsilon,\delta)$-differential privacy} if for all pairs of neighboring inputs $\mathcal{I},\mathcal{I}'$ and for all subsets $\mathcal{S}\subset\mathcal{R}$ it holds that
\[\Pr[\mathsf{ALG}(\mathcal{I}) \in \mathcal{S}] \leq e^\epsilon \Pr[\mathsf{ALG}(\mathcal{I'}) \in \mathcal{S}] + \delta\]
\end{definition}

For each type of input we are interested in answering a different type of query. For graphs, we are interesting in \emph{cut-queries}: given a nonempty subset $S$ of the vertices of the graph, we wish to know what is the total weight of edges crossing the $(S,\bar S)$-cut. We denote this as $\Phi_G(S) = \sum_{u\in S, v\notin S} w_{u,v}$. 
\begin{definition}
We say an algorithm $\textsf{ALG}$ gives a $(\eta,\tau,\nu)$-approximation for cut queries, if for every nonempty $S$ it holds that 
\[\Pr\left[(1-\eta)\Phi_G(S) - \tau \leq \mathsf{ALG}(S) \leq (1+\eta)\Phi_G(S) + \tau\right] \geq 1 - \nu\]
\end{definition}

For $n\times d$ matrices, we are interested in \emph{directional variance queries}: given a unit-length direction $x$, we wish to know what's the variance of $A$ along the $x$ direction: $\Phi_A(x) = x\T A\T A x$. (Our algorithm normalizes $A$ s.t. the mean of its $n$ rows is $0$.)
\begin{definition}
We say an algorithm $\textsf{ALG}$ gives a $(\eta,\tau,\nu)$-approximation for directional variance queries, if for every unit-length vector $x$ it holds that 
\[\Pr\left[(1-\eta)\Phi_A(x) - \tau \leq \mathsf{ALG}(x) \leq (1+\eta)\Phi_A(x) + \tau\right] \geq 1 - \nu\]
\end{definition}

\paragraph{Some Linear Algebra.} Given a $m\times n$ matrix $M$ its Singular Value Decomposition (SVD) is $M = U \Sigma V\T$ where $U\in \mathbb{R}^{m\times m}$ and $V\in\mathbb{R}^{n\times n}$ are unitary matrices, and $\Sigma$ has non-zero values only on its main diagonal. Furthermore, there are exactly $rank(M)$ positive values on the main diagonal, denoted $\sigma_1(M) \geq \ldots \geq \sigma_{rank(M)}(M)$, called the \emph{singular values}. This allows us to write $M$ as the sum of $rank(M)$ rank-$1$ matrices: $M = \sum_{i=1}^{rank(M)} \sigma_i u_i v_i\T$. Because $\Sigma$ has non-zero values only on its main diagonal, the notation $\Sigma^i$ denotes a matrix whose non-zero values lie only on the main diagonal and are $\sigma_1^i(M), \sigma_2^i(M), \ldots,\sigma_{rank(M)}^i(M)$. Using the SVD, it is clear that if $M$ is of full-rank, then $M^{-1} = V\Sigma^{-1} U\T$, and that if $n=m=rank(M)$ then $\det(M) = \prod_{i=1}^{n}\sigma_i(M)$. Furthermore, even when $M$ is not full-rank, the SVD allows us to use similar notation to denote the generalizations of the inverse and of the determinant: The Moore-Penrose inverse of $M$ is $M\mpinv = V \Sigma^{-1} U\T$; and the pseudo-determinant of $M$ is $\tilde\det(M) = \prod_{i=1}^{rank(M)} \sigma_i(M)$. A $n\times n$ symmetric matrix is called \emph{positive semidefinite} (PSD) if it holds that $x\T M x\geq 0$ for every $x\in\mathbb{R}^n$. Given two PSDs $M$ and $N$ we denote the fact that $(N-M)$ is PSD by $M\preceq N$. For further details, see~\cite{HornJohnson}. 

\paragraph{Gaussian distribution.} Given a r.v. $X$, we denote by $X\sim\mathcal{N}(\mu,\sigma^2)$ the fact that $X$ has normal distribution with  mean $\mu$ and variance $\sigma^2$. Recall that $\PDF_X(x) = \frac {1} {\sqrt{2\pi\sigma^2}} \exp(-(x-\mu)^2/2\sigma^2)$. We repeatedly apply the \emph{linear combination} rule: for any two i.i.d normal random variables s.t. $X\sim \mathcal{N}(\mu_X, \sigma^2_X)$ and $Y\sim \mathcal{N}(\mu_Y, \sigma^2_Y)$, we have that their linear combination $Z = aX+bY$ is distributed according to $Z\sim \mathcal{N}(a\mu_X + b\mu_Y, a^2\sigma_X^2 + b^2\sigma_Y^2)$. This in turn allows us to identify a random variable $R\sim\mathcal{N}(0,\sigma^2)$ with the random variable $\sigma R'$, where $R'\sim\mathcal{N}(0,1)$. Classic concentration bounds on Gaussians give that $Pr[|x-\mu|^2 > \log(1/\delta)\sigma^2] \leq 2\delta$. 

The multivariate normal distribution is the multi-dimension extension of the univariate normal distribution. $X\sim \mathcal{N}(\mu, \Sigma)$ denotes a $m$-dimensional multivariate r.v. whose mean is $\mu \in \mathbb{R}^m$, and variance is the PSD matrix $\Sigma = \E\left[(X-\mu)(X-\mu)\T\right]$. If $\Sigma$ has full rank ($\Sigma$ is positive definite) then $\PDF_X(x) = \frac 1 {\sqrt{ (2\pi)^m \det(\Sigma) }}\exp(-\frac 1 2 x\T \Sigma^{-1} x)$, a well defined function. If $\Sigma$ has non-trivial kernel space then $\PDF_X$ is technically undefined (since $X$ is defined only on a subspace of volume $0$, yet $\int_{\mathbb{R}^m}\PDF_X(x)dx = 1$). However, if we restrict ourselves only to the subspace $\mathcal{V} = (Ker(\Sigma))^{\perp}$, then $\PDF_X^{\mathcal{V}}$ is defined over $\mathcal{V}$ and $\PDF_X^{\mathcal{V}}(x) = \frac 1 {\sqrt{(2\pi)^{rank(\Sigma)} \tilde\det(\Sigma)}} \exp(-\frac 1 2 x\T\Sigma\mpinv x)$. From now on, we omit the superscript from the $\PDF$ and refer to the above function as the $\PDF$ of $X$. Observe that using the SVD, we can denote $\Sigma = U~\textrm{diag}(\sigma_1^2, \sigma_2^2, \ldots, \sigma_r^2,0, \ldots, 0) ~U\T$, and so $\mathcal{V}$ is the subspace spanned by the first $r$ rows of $U$. The multivariate extension of the linear combination rule is as follows. If $A$ is a $n\times m$ matrix, then the multivariate r.v. $Y = AX$ is distributed as though $Y\sim\mathcal{N}(A\mu, A \Sigma A\T)$. For further details regarding multivariate Gaussians see~\cite{miller1964multidimensional}.

Finally, we conclude these Gaussian preliminaries with the famous Johnson-Lindenstrauss Lemma, our main tool in this paper.
\begin{theorem}[The Johnson Lindenstrauss transform~\cite{JL84}]
\label{thm:JL}
Fix any $0< \eta < 1/2$. Let $M$ be a $r\times m$ matrix whose entries are iid samples from $\mathcal{N}(0,1)$. Then $\forall x \in \mathbb{R}^{m}$. 
\[\Pr_M\left[ (1-\eta)\|x\|^2 \leq \frac 1 r \|Mx\|^2 \leq (1+\eta)\|x\|^2\right] \geq 1 - 2\exp(-\eta^2r/8)\]
\end{theorem}

\paragraph{Laplacians and edge-matrices.} An undirected weighted graph $G = (V(G),E(G))$ can be represented in various ways. One representation is by the \emph{adjacency matrix} $A$, where $A_{u,v} = w_{u,v}$. Another way is by the $\binom n 2 \times n$ \emph{edge matrix} of the graph, $E_G$. We assume that the vertices of $G$ are ordered arbitrarily, and for each pair of vertices $\{u,v\}$ where $u<v$, there exists a row in $E_G$. The entries of $E_G$ are
\[ \big(E_G\big)_{(\{u,v\}, x)} = \left\{ \sqrt{w_{u,v}}, \textrm{ if } u\sim_G v\textrm{ and } x=u\ ; \qquad -\sqrt{w_{u,v}}, \textrm{ if } u\sim_G v\textrm{ and } x=v\ ; \qquad 0, \textrm{ o/w}\right\}\] 
where $u\sim_G v$ denotes that $(u,v)$ is an edge in $G$. Alternatively, one can represent $G$ using the \emph{Laplacian} of the graph $L_G = E_G\T E_G$. Formally, the matrix $L_G$ is the matrix whose diagonal entries are $(L_G)_{u,u} = \sum_{x \sim_G u} w_{x,u}$ and non diagonal entries are $(L_G)_{u,v} = -w_{u,v}$.
It is simple to verify that for any $x$, the following equality holds: $x\T L_G x = \sum_{u \sim_G v} w_{u,v}(x_u-x_v)^2$. As a corollary, if we take any nonempty $S\subsetneq V(G)$ and denote its $\bitsn$-indicator vector as $\1_S$,
then $\1_S\T~L_G~\1_S = \|E_G\1_S\|^2 = \sum_{u\in S, v\notin S} w_{u,v} = \Phi_G(S)$.

\paragraph{Additional notations.} We denote by $e_a$ the indicator vector of $a$. We denote by $e_{a,b} = e_a - e_b$. It follows that the $n\times n$ matrix $L_{a,b} = e_{a,b} e_{a,b}\T$ is the matrix whose projection over coordinates $a,b$ is $\begin{pmatrix}1& -1\cr -1 & 1\cr\end{pmatrix}$, while every other entry is $0$. We also denote $E_{a,b}$ as the $\binom n 2\times n$ matrix, whose rows are all zeros except for the row indexed by the $(a,b)$ pair, which is $e_{a,b}\T$. Observe: $L_{a,b} = e_{a,b} e_{a,b}\T = E_{a,b}\T E_{a,b}$.

\section{Publishing a Perturbed Laplacian}
\label{sec:perturbed_laplacian}

\subsection{The Johnson-Lindenstrauss Algorithm}
\label{subsec:JL-algo}
We now show that the Johnson Lindenstrauss transform preserves differential privacy. We first detail our algorithm, then analyze it.

\begin{algorithm}[H]
\DontPrintSemicolon
\KwIn{A $n$-node graph $G$, parameters: $\epsilon, \delta, \eta,\nu > 0$}
\KwOut{A Laplacian of a graph $\tilde L$}
\BlankLine
Set $r = \frac {8 \ln(2/\nu)}{\eta^2}$, and $w = \frac {\sqrt {32r\ln(2/\delta)}} {\epsilon} \ln(4r/\delta)$ \;
For every $u\neq v$, set $w_{u,v} \leftarrow \frac w n + \left(1-\frac w n\right)w_{u,v}$. \;
Pick a matrix $M$ of size $r\times \binom n 2$, whose entries are iid samples of $\mathcal{N}(0,1)$.\;
\KwRet{$\tilde L = \frac 1 r E_{G}\T M\T ME_{G}$}\;

\caption{Outputting the Laplacian of a Graph while Preserving Differential Privacy\label{alg:output_perturbed_laplacian}}
\end{algorithm}

\begin{algorithm}[H]
\DontPrintSemicolon
\LinesNotNumbered
\KwIn{A non empty $S\subsetneq V(G)$, parameters $n$, $w$ and Laplacian $\tilde L$ from Algorithm~\ref{alg:output_perturbed_laplacian}.}
\BlankLine
\KwRet{$R(S) = \frac 1 {1-\frac w n} \left(\1_S\T \tilde L \1_S - w \frac {s(n-s)}n\right)$}\;
\caption{Approximating $\Phi_G(S)$\label{alg:approx_cut}}
\end{algorithm}

\begin{theorem}
\label{thm:perturbed_laplacian_preserves_privacy}
Algorithm~\ref{alg:output_perturbed_laplacian} preserves $(\epsilon,\delta)$-differential privacy w.r.t to edge changes in $G$.
\end{theorem}

\begin{theorem}
\label{thm:approx_cut}
For every $\eta,\nu>0$ and a nonempty $S$ of size $s$, Algorithm~\ref{alg:approx_cut} gives a $\left(\eta, \tau, \nu \right)$-approximation for cut queries, for $\tau = O\big ( s\cdot \frac {\sqrt{\ln(1/\delta)\ln(1/\nu)}} \epsilon \left(\ln(1/\delta) + \ln(\ln(1/\nu)/\eta^2)\right) \big )$.
\end{theorem}

Clearly, once Algorithm~\ref{alg:output_perturbed_laplacian} publishes $\tilde L$, any user interested in estimating $\Phi_G(S)$ for some nonempty $S\subsetneq V(G)$ can run Algorithm~\ref{alg:approx_cut} on her own. Also, observe that $w$ is independent of $n$, which we think of as large number, so we assume thoughout the proofs of both theorems that both $\frac w n, \frac 1 w$ are $ < 1/2$. 
\ifx\ConferenceVersion\undefined
Now, the proof of Theorem~\ref{thm:approx_cut} is immediate from the JL Lemma.

\input{util_proof_laplacian}
\else
The proof of Theorem~\ref{thm:approx_cut} is immediate from the JL lemma. It is deferred to Appendix~\ref{sec:apx_laplacians_completions}, along with the comparison of Algorithm~\ref{alg:output_perturbed_laplacian} with existing algorithms.
\fi

\paragraph{Comment.} The guarantee of Theorem~\ref{thm:approx_cut} is not to be mistaken with a weaker guarantee of providing a good approximation to \emph{most} cut-queries. Theorem~\ref{thm:approx_cut} guarantees that any set of $k$ predetermined cuts is well-approximated by Algorithm~\ref{alg:approx_cut}, assuming Algorithm~\ref{alg:output_perturbed_laplacian} sets $\nu < 1/2k$. In contrast, giving a good approximation to most cuts can be done by a very simple (and privacy preserving) algorithm: by outputting the number of edges in the graph (with small Laplacian noise). Afterall, we expect a cut to have $\frac {m} {\binom n 2} s(n-s)$ edges crossing it.

We turn our attention to the proof of Theorem~\ref{thm:perturbed_laplacian_preserves_privacy}.
We fix any two graphs $G$ and $G'$, which differ only on a single edge, $(a,b)$. We think of $(a,b)$ as an edge in $G'$ which isn't present in $G$, and in the proof of Theorem~\ref{thm:perturbed_laplacian_preserves_privacy}, we identify $G$ with the manipulation Algorithm~\ref{alg:output_perturbed_laplacian} performs over $G$, and assume that the edge $(a,b)$ is present in both graphs, only it has weight $\frac w n$ in $G$, and weight $1$ in $G'$. Clearly, this analysis carries on for a smaller change, when the edge $(a,b)$ is present in both graphs but with different weights. (Recall, we assume all edge weights are bounded by $1$.)

Now, the proof follows from assuming that Algorithm~\ref{alg:output_perturbed_laplacian} outputs the matrix $O = ME_G$, instead of $\tilde L = \frac 1 r O\T O$. (Clearly, outputting $O$ allows one to reconstruct $\tilde L$.) Observe that $O$ is composed of $r$ identically distributed rows: each row is created by sampling a $\binom n 2$-dimensional vector $Y$ whose entries $\sim\mathcal{N}(0,1)$, then outputting $Y\T E_G$. Therefore, we prove Theorem~\ref{thm:perturbed_laplacian_preserves_privacy} by showing that each row maintain $(\epsilon_0, \delta_0)$-differential privacy, for the right parameters $\epsilon_0,\delta_0$. To match standard notion, we transpose row vectors to column vectors, and compare the distributions $E_G\T Y$ and $E_{G'}\T Y$. 

\begin{claim} 
\label{clm:inequalities_of_privacy}
Set $\epsilon_0 = \frac {\epsilon} {\sqrt{4r\ln(2/\delta)}}, \delta_ 0 = \frac {\delta}{2r}$. Then, 
\begin{equation}
\label{eq:privacy_UB} 
\forall x,\ \  \PDF_{E_G\T Y}(x) \leq~~ e^{\epsilon_0} \PDF_{E_{G'}\T Y}(x)
\end{equation}
Denote $S = \{x:\ \PDF_{E_G\T Y}(x) ~\geq~~ e^{-\epsilon_0}\PDF_{E_{G'}\T Y}(x)\}$. Then
\begin{equation}
\label{eq:privacy_LB}  
\Pr[S] \geq 1- \delta_0
\end{equation}
\end{claim}
\begin{proof}[Proof of Theorem~\ref{thm:perturbed_laplacian_preserves_privacy} based on Claim~\ref{clm:inequalities_of_privacy}]
Apply the composition theorem of~\cite{DRV10} for $r$ iid samples each preserving $(\epsilon_0,\delta_0)$-differential privacy.
\end{proof}

To prove Claim~\ref{clm:inequalities_of_privacy}, we denote $X = E_G\T Y$ and $X' = E_{G'}\T Y$. From the preliminaries it follows that $X$ is a multivariate Gaussian distributed according to $\mathcal{N}(0, E_G\T I_{\scriptscriptstyle \binom n 2 \times \binom n 2} E_G) = \mathcal{N}(0, L_G)$, and similarly, $X'\sim\mathcal{N}(0, L_{G'})$. In order to analyze the two distributions, $\mathcal{N}(0, L_G)$ and $\mathcal{N}(0, L_{G'})$, we now discuss several of the properties of $L_G$ and $L_{G'}$, then turn to the proof of Claim~\ref{clm:inequalities_of_privacy}.

First, it is clear from definition that the all ones vector, $\mathbf{1}$, belongs to the kernel space of $E_G$ and $E_{G'}$, and therefore to the kernel space of $L_G$ and $L_{G'}$. Next, we establish a simple fact.
\begin{fact}
\label{fact:1_is_only_kernel}
If $G$ is a graph s.t. for every $u\neq v$ we have that $w_{u,v}>0$, then $\mathbf{1}$ is the only vector in the kernel space of $E_G$ and $L_G$.
\end{fact}
\begin{proof}
Any non-zero $x\perp\mathbf{1}$ has at least one positive coordinate and one negative coordinate, thus the non-negative sum $\|E_Gx\|^2 = x\T L_G x = \sum_{u\neq v} w_{u,v}(x_u-x_v)^2$ is strictly positive.
\end{proof}

Therefore, the kernel space of both $L_G$ and of $L_{G'}$ is exactly the $1$-dimensional span of the $\mathbf{1}$ vector (for every possible outcome $y$ of $Y$ we have that $E_G\T y \cdot \mathbf{1} = E_{G'}\T y \cdot \mathbf{1} = 0$). Alternatively, both $X$ and $X'$ have support which is exactly $\mathcal{V} = \mathbf{1}^\perp$. Hence, we only need to prove the inequalities of Claim~\ref{clm:inequalities_of_privacy} for $x\in\mathcal{V}$. Secondly, observe that $L_{G'} = L_G + (1-\frac w n) L_{a,b}$. Therefore, it holds that for every $x\in \mathbb{R}^n$ we have $x\T L_{G'} x = x\T L_G x + (1-\frac w n)(x_a-x_b)^2 \geq x\T L_G x$. In other words, $L_{G} \preceq L_{G'}$, a fact that yields several important corollaries.

We now introduce notation for the Singular Value Decomposition of both $L_G$ and $L_{G'}$. We denote $E_G\T = U \Sigma V\T$ and $E_{G'}{\T} = U' \Lambda V'{\T}$, resulting in $L_G = U \Sigma^2 U\T, L_{G'} = U' \Lambda^2 U'{\T}, L_{G}\mpinv = U \Sigma^{-2} U\T$ and $L_{G'}\mpinv = U' \Lambda^{-2} U'{\T}$. We denote the singular values of $L_G$ as $\sigma_1^2 \geq \ldots \geq \sigma_{n-1}^2 >  \sigma_n^2=0$, and the singular values of $L_{G'}$ as $\lambda_1^2 \geq \ldots \geq \lambda_{n-1}^2 > \lambda_n^2=0$. Weyl's inequality allows us to deduce the following fact. Its and other facts' proofs are in Appendix~\ref{sec:apx_linear_algebra}.
\begin{fact}
\label{fact:Weyl_inequality}
Since $L_{G}\preceq L_{G'}$ then for every $i$ we have that $\lambda_i^2 \geq \sigma_i^2$.
\end{fact}
In addition, since Algorithm~\ref{alg:output_perturbed_laplacian} alters the input graphs s.t. the complete graph $\frac w n L_{K_n}$ is contained in $G$, then it also holds that $\frac w n L_{K_n}\preceq L_G$, and so Fact~\ref{fact:Weyl_inequality}  gives that for every $1\leq i \leq n-1$ we have that $\sigma_i^2 \geq w = \frac w n \cdot n$.  (It is simple to see that the eigenvalues of $K_n$ are $\{n,n,\ldots, n, 0\}$.) Furthermore, as $L_{G'} = L_G + (1-\frac w n)L_{a,b}$ and the singular values of $L_{a,b}$ are $\{2,0,0,\ldots,0\}$, then we have that 
\[\sum_i \lambda_i^2 = tr(L_{G'}) \leq tr(L_G) + tr\left((1-\frac w n)L_{a,b}\right) \leq \sum_i\sigma_i^2 + 2\]

Another fact we can deduce from $L_G \preceq L_{G'}$, is the following.
\begin{fact}
\label{fact:inverse_order}
Since the kernels of $L_G$ and of $L_{G'}$ are identical, then for every $x$ it holds that $x\T L_{G'}\mpinv x \leq x\T L_G\mpinv x$. Symbolically, $L_{G} \preceq L_{G'} \Rightarrow L_{G'}\mpinv \preceq L_G\mpinv$.
\end{fact}

Having established the above facts, we can turn to the proof of privacy.

\begin{proof}[Proof of Claim~\ref{clm:inequalities_of_privacy}]
We first prove the upper bound in~\eqref{eq:privacy_UB}. As mentioned, we focus only on $x\in \mathcal{V} = \mathbf{1}^\perp$, where 
\begin{eqnarray*}
&&\PDF_{E_G\T Y}(x) = \left((2\pi)^{n-1} \tilde\det(L_G)\right)^{-1/2} \exp(-\frac 12 x\T L_G\mpinv x) \cr
&&\PDF_{E_{G'}\T Y}(x) =  \left((2\pi)^{n-1} \tilde\det(L_{G'})\right)^{-1/2}\exp(-\frac 12 x\T L_{G'}\mpinv x)\cr
\end{eqnarray*}
As noted above, we have that for every $x$ it holds that $x\T L_{G'}\mpinv x \leq x\T L_G\mpinv x$, so $\exp(-\frac 12 x\T L_G\mpinv x) \leq \exp(-\frac 12 x\T L_{G'}\mpinv x)$. It follows that for every $x$ we have that $\frac {\PDF_{E_G\T Y}(x)} {\PDF_{E_{G'}\T Y}(x)} \leq \left(\frac {\tilde\det (L_{G'})} {\tilde\det (L_G)} \right)^{1/2} = \left(\prod_{i=1}^{n-1} \frac {\lambda_i^2} {\sigma_i^2}\right)^{1/2} $.
Denoting $\Delta_i = \lambda^2_i - \sigma_i^2\geq 0$, and recalling that $\sum_i \Delta_i \leq 2$ and that $\forall i, \sigma_i^2 \geq w$ it holds that
\[  \frac {\PDF_{E_G\T Y}(x)} {\PDF_{E_{G'}\T Y}(x)} \leq \sqrt{ \prod_{i=1}^{n-1} \left(1+\frac {\Delta_i} {\sigma_i^2} \right)}\leq \exp\left(\frac 1 {2w} \sum_i \Delta_i \right)\leq e^{\frac 1 w} \leq e^{\frac {\epsilon} {\sqrt{4r\ln(2/\delta)}}} = e^{\epsilon_0}\]

We now turn to the lower bound of~\eqref{eq:privacy_LB}. We start with analyzing the term $x\T L_{G}\mpinv x$ that appears in $\PDF_{E\T Y}(x)$. Again, we emphasize that $x\in \mathcal{V}$, justifying the very first equality below.
\begin{eqnarray*}
& x\T L_{G}\mpinv x & = x\T L_G\mpinv L_{G'} L_{G'}\mpinv x = x\T L_G\mpinv \left(L_{G}+(1-\frac w n) L_{ab}\right) L_{G'}\mpinv x \cr
&& = x\T L_{G'}\mpinv x\  + \  (1-\frac w n) x\T L_G\mpinv L_{a,b} L_{G'}\mpinv x \cr
&& = x\T L_{G'}\mpinv x\  + \ \ (1-\frac w n) x\T L_G\mpinv e_{a,b} ~ \cdot ~ e_{a,b}\T L_{G'}\mpinv x \cr
\end{eqnarray*}
Therefore, if we show that
\begin{equation}
\Pr _{x\sim E_G\T Y} \left[ x\T L_G\mpinv e_{a,b} ~ \cdot ~ e_{a,b}\T L_{G'}\mpinv x > \frac 2 {1-\frac w n}  \epsilon_0 \right] < \delta_0 \label{eq:bad_event}
\end{equation}
then it holds that w.p. $> 1- \delta_0$ we have
\[ \frac {\PDF_{E_G\T Y}(x)} {\PDF_{E_{G'}\T Y}(x)} \geq 1\cdot  \exp\left(-\frac 1 2 x\T (L_G\mpinv-L_{G'}\mpinv)x\right) \geq \exp\left(- \frac {1-\frac w n} 2   x\T L_G\mpinv e_{a,b} ~ \cdot ~ e_{a,b}\T L_{G'}\mpinv x  \right) \geq e^{-\epsilon_0}\]
which proves the lower bound of~\eqref{eq:privacy_LB}. We turn to proving~\eqref{eq:bad_event}.

Denote $term_1 = e_{a,b}\T L_G\mpinv x$ and $term_2 = e_{a,b}\T L_{G'}\mpinv x$. Since $x = E_G\T y$ where $y\sim Y$ then $term_i$ is distributed like $vec_i\T Y$ where $vec_1 = E_G L_G\mpinv e_{a,b}$ and $vec_2 = E_G L_{G'}\mpinv e_{a,b}$. The na\"ive bound, $\|vec_1\| \leq \|E_G\|~\|L_G\mpinv\| \|e_{a,b}\|$ gives a bound on the size of $vec_1$ which is dependent on the ratio $\frac {\sigma_1}{\sigma_{n-1}^2}$. We can improve the bound, on both $\|vec_1\|$ and $\|vec_2\|$, using the SVD of $E_G$ and $E_{G'}$.
\begin{eqnarray*}
& \|vec_1\| & = \|E_G L_G\mpinv e_{a,b}\| = \|V\Sigma U\T U \Sigma^{-2} U\T e_{a,b}\| = \|V \Sigma^{-1} U\T e_{a,b}\| \cr
&& \leq \|V\|~ \|\Sigma^{-1}\| ~\|U\|~ \|e_{a,b}\| = 1\cdot \sigma_{n-1}^{-1} \cdot 1 \cdot \sqrt 2 = \frac {\sqrt 2} {\sqrt w}\cr
& \|vec_2\| & = \|E_G L_{G'}\mpinv e_{a,b}\| = \|(E_{G'} - (1-\frac w n)E_{a,b}) L_{G'}\mpinv e_{a,b}\| < \|E_{G'}L_{G'}\mpinv e_{a,b}\|  + \|E_{a,b}L_{G'}\mpinv e_{a,b}\| \cr
&& \stackrel {(\ast)} \leq \lambda_{n-1}^{-1} \cdot \sqrt 2 + \|E_{a,b}L_{G'}\mpinv e_{a,b}\| \stackrel {(\ast \ast)} = \frac {\sqrt 2}{\sqrt w} + e_{a,b}\T L_{G'}\mpinv e_{a,b} \cr
&& \leq \frac {\sqrt 2} {\sqrt w} + \frac 2 w = \frac {\sqrt 2} {\sqrt w}\left(1+\frac {\sqrt 2} {\sqrt w}\right)
\end{eqnarray*}
where the bound in $(\ast)$ is derived just like in $vec_1$ (using $E_{G'}L_{G'}\mpinv e_{a,b} = V'\Lambda U'{\T} U' \Lambda^{-2} U'{\T}e_{a,b}$) , and the equality in $(\ast \ast)$ follows from the fact that all coordinates in the vector $E_{a,b}L_{G'}\mpinv e_{a,b}$ are zero, except for the coordinate indexed by the $(a,b)$ pair.

We now use the fact that $term_1$ and $term_2$ are both linear combinations of  i.i.d $\mathcal{N}(0,1)$ random variables. Therefore for $i=1,2$ we have that $term_i \sim \mathcal{N} (0, \|vec_i\|^2)$ so $\Pr[|term_i| > \sqrt{\log(2/\delta_0)}\|vec_i\|] \leq e^{-\frac{\|vec_i\|^2\log(2/\delta_0)}{\|vec_i\|^2}} < \frac{\delta_0} 2$. It follows that w.p $>1- \delta_0$ both $|term_1| <\sqrt{\log(2/\delta_0)}\sqrt{\frac 2 w}$ and $|term_2| \leq \sqrt{\log(2/\delta_0)}\sqrt{\frac 4 w}$, so $term_1\cdot term_2 \leq \sqrt 8 \log(2/\delta_0) /w$. Plugging in the value of $w$, we have that $\Pr[term_1\cdot term_2 \leq 2\epsilon_0] \geq 1-\delta_0$ which concludes the proof of~\eqref{eq:bad_event} and of Claim~\ref{clm:inequalities_of_privacy}.
\end{proof}

\ifx\ConferenceVersion\undefined
\input{comparison}
\else
\fi

\section{Publishing a Covariance Matrix}
\label{sec:covariance_matrix}

\ifx\ConferenceVersion\undefined
\subsection{The Algorithm}
\label{subsec:algorithm_covariance}
\else
\fi

In this section, we are concerned with the question of allowing users to estimate the covariance of a given sample data along an arbitrary direction $x$. 
We think of our input as a $n\times d$ matrix $A$, and we maintain privacy w.r.t to changing the coordinates of a single row s.t. a vector $v$ of size $1$ is added to $A_{(i)}$. We now detail our algorithm for publishing the covariance matrix of $A$. Observe that in addition to the variance, we can output $\mu = \frac 1 n A\T \1$, the mean of all samples in $A$, in a differentially private manner by adding random Gaussian noise. (We merely output $\tilde \mu = \mu + \mathcal{N}(0, \frac {4\log(1/\delta)} {n^2\epsilon^2} I_{d\times d})$.) We denote by $I_{n\times d}$ the $n\times d$ matrix whose main diagonal has $1$ in each coordinate and all other coordinates are $0$.
\ifx\ConferenceVersion\undefined
\else
We detail the algorithms here, but prove their privacy and utility in Appendix~\ref{sec:apx_covariance_completions}, along with comparing Algorithm~\ref{alg:perturbed_covariance} to existing algorithms.
\fi

\begin{algorithm}[H]
\DontPrintSemicolon
\KwIn{A $n\times d$ matrix $A$. Parameters $\epsilon,\delta,\eta,\nu >0$.}
\BlankLine
Set $r = \frac {8\ln(2/\nu)} {\eta^2}$ and $w = \frac {16\sqrt{r\ln(2/\delta)}} \epsilon \ln(16r/\delta)$.\;
Subtract the mean from $A$ by computing $A \leftarrow A - \frac 1 n\1\1^T A$.\;
Compute the SVD of $A = U\Sigma V\T$.\;
Set $A \leftarrow U(\sqrt{\Sigma^2 + w^2I_{n\times d}}) V\T$.\;
Pick a matrix $M$ of size $r\times n$ whose entries are iid samples of $\mathcal{N}(0,1)$.\;
\KwRet{$\tilde C = \frac 1 r A\T M\T M A$}.
\caption{Outputting a Covariance Matrix while Preserving Differential Privacy\label{alg:perturbed_covariance}}
\end{algorithm}

\begin{algorithm}[H]
\DontPrintSemicolon
\LinesNotNumbered
\KwIn{A unit-length vector $x$, parameter $w$ and a Covariance matrix $\tilde C$ from Algorithm~\ref{alg:perturbed_covariance}.}
\BlankLine
\KwRet{$R(x) = x\T \tilde C x - w^2$}.
\caption{Approximating $\Phi_A(x)$\label{alg:approx_directional_variance}}
\end{algorithm}

\begin{theorem}
\label{thm:perturbed_covariance_preserves_privacy}
Algorithm~\ref{alg:perturbed_covariance} preserves $(\epsilon,\delta)$-differential privacy.
\end{theorem}

\begin{theorem}
\label{thm:utility_directional_variance}
Algorithm~\ref{alg:approx_directional_variance} is a $(\eta,\tau,\nu)$-approximation for directional variance queries, where $\tau = O\left( \frac {\ln(1/\delta)\ln(1/\nu)} {\epsilon^2\eta} \ln^2\left( \frac {\ln(1/\nu)}{\delta \eta^2} \right) \right)$.
\end{theorem}
\ifx\ConferenceVersion\undefined
\input{util_proof_var}
\else
\fi

\paragraph{Comment.} We wish to clarify that Theorem~\ref{thm:utility_directional_variance} does \emph{not} mean that we publish a matrix $\tilde C$ which is a low-rank approximation to $A\T A$. It is also not a matrix on which one can compute an approximated PCA of $A$, \emph{even if we set} $\nu = 1/\poly(d)$. The matrix $\tilde C$ should be thought of as a ``test-matrix'' -- if you believe $A$ has high directional variance along some direction $x$ then you can test your hypothesis on $\tilde C$ and (w.h.p) get the good approximated answer. However, we do not guarantee that the singular values of $A\T A$ and of $\tilde C$ are close or that the eigenvectors of $A\T A$ and $\tilde C$ are comparable. (See discussion in Section~\ref{sec:open_problems}.)

\ifx\ConferenceVersion\undefined
\input{privacy_proof_var}
\else
\fi 
\paragraph{Corollary.} Using the definitions of $r$ and $w$ as in Algorithm~\ref{alg:perturbed_covariance} -- the proof of Theorem~\ref{thm:perturbed_covariance_preserves_privacy} actually shows that in the case that $A$ is a matrix with all singular values $\geq w$, then the following simple algorithm preserves $(\epsilon,\delta)$-differential privacy: pick a random $r\times n$ matrix $M$ whose entries are iid normal Gaussians, and output $O = MA$. Furthermore, observe that if $\sigma_d$, the least singular value of $A$, is bigger than, say, $10w$, then one can release $\sigma_d + Lap(1/\epsilon)$ then release $O=MA$. In such a case, users know that for any unit vector $x$ w.p. $\geq 1-\nu$ it holds that $\tfrac 1 r\|Ox\|^2 \leq (1\pm\eta)\|Ax\|^2$.

\paragraph{Comment.} Comparing Algorithms~\ref{alg:output_perturbed_laplacian} and~\ref{alg:perturbed_covariance}, we have that in $L_G = E_G\T E_G$ we ``translate'' the spectral values by $w$, and in $A\T A$ we ``translated'' the spectral values by $w^2$. This is an artifact of the ability to directly compare  the spectal values of $L_G$ and $L_{G'}$ in the first analysis, whereas in the second analysis we compare the spectral values of $A$ and $A'$ (vs. $A\T A$ and ${A'}\T A'$). This is why the noise bounds in the general case are $\tilde O(1/\epsilon\eta)$ times worse than for graphs.

\ifx\ConferenceVersion\undefined
\input{comparison_var}
\else
\fi

\section{Discussion and Open Problems}
\label{sec:open_problems}

The fact that the JL transform preserves differential privacy is likely to have more theoretical and practical applications than the ones detailed in this paper. Below we detail a few of the open questions we find most compelling.

\paragraph{Error depedency on $r$.} Our algorithm projects the edge-matrix of a given graph on $r$ random directions, then publishes these projections. The value of $r$ determines the probability we give a good approximation to a given cut-query, and provided that we wish to give a good approximation to all cut-queries, our analysis requires us to set $r = \Omega(n)$. But is it just an artifact of the analysis? Could it be that a better analysis gives a better bound on $r$? It turns out that the answer is  ``no''. In fact, the direction on which we project the data now have high correlation with the published Laplacian. We demonstrate this with an example.

Assume our graph is composed of a single perfect matching between $2n$ nodes, where node $i$ is matched with node $n+i$. Focus on a single random projection -- it is chosen by picking $\binom {2n} 2$ iid random values $x_{i,j}\sim \mathcal{N}(0,1)$, and for the ease of exposition imagine that the values of the edges in the matching are picked first, then the values of all other pairs of vertices. Now, if we pick the value $x_{i,n+i}$ for the $\langle i, n+i\rangle$ edge, then node $i$ is assigned $x_{i,n+i}$ while node $n+i$ is assigned $-x_{i,n+i}$. So regardless of the sign of $x_{i,n+i}$, \emph{exactly one} of the two nodes $\{i,n+i\}$ is assigned the positive value $|x_{i,n+i}|$ and exactly one is assigned the negative value $-|x_{i,n+i}|$. Define $S$ as the set of $n$ nodes that are assigned the positive values and $\bar S$ as the set of $n$ nodes that are assigned the negative values. The sum of weight crossing the $(S,\bar S)$-cut is distributed like $(X + \frac w n Y)^2$ where $X = \sum_i |x_{i,n+i}|$ and $Y = \sum_i \sum_{j\neq n+i} x_{i,j}$. Indeed, $Y$ is the sum of $n(n-1)$ random normal iid Gaussians, but $X$ is the sum of $n$ \emph{absolute values} of Gaussians. So w.h.p. both $X$ and $Y$ are proportional to $n$. Therefore, in the direction of this particular random projection we estimate the $(S,\bar S)$-cut as $\Omega((n \pm w)^2)=\Omega(n^2)$ rather than $O(n)$. (If $X$ was distributed like the sum of $n$ iid normal Gaussians, then the estimation would be proportional to $(\sqrt n)^2 = n$.)

Assuming that the remaining $r-1$ projections estimate the cut as $O(n)$, then by averaging over all $r$ random projections our estimation of the $(S,\bar S)$-cut is $\omega(n)$, as long as $r = o(n)$.

\paragraph{Error amplification or error detection.} Having established that we do err on some cuts, we pose the question of error amplification. Can we introduce some error-correction scheme to the problem without increasing $r$ significantly? Error amplification without increasing $r$ will allow us to keep the additive error fairly small. One can view $\tilde L$ as a coding of answers to all $2^n$ cut-queries which is guaranteed to have at least $1-\nu$ fraction of the code correct, in the sense that we get a $(\eta,\tau)$-approximation to the true cut-query answer. As such, it is tempting to try some self-correcting scheme -- like adding a random vector $x$ to the vector $\1_S$, then finding the estimation to $x\T L_G x$ and $(\1_s + x)\T L_G x$ and inferring $\1_S\T L_G \1_S$. We were unable to prove such scheme works due to the dot-product problem (see next paragraph) and to query dependencies.

A related question is of error detection: can we tell whether $\tilde L$ gives a good estimation to a cut query or not? One potential avenue is to utilize the trivial guess for $\Phi_G(S)$ -- the expected value $\frac m {\binom n 2}{s(n-s)}$ (we can release $m$ via the Laplace mechanism). We believe this question is related to the problem of estimating the variance of $\{\Phi_G(S)~ : \ |S| = s\}$.



\paragraph{Edges between $S$ and $T$.} Our work assures utility only for cut-queries. It gives no utility guarantees for queries regarding $E(S,T)$, the set of edges connecting two disjoint vertex-subsets $S$ and $T$. The reason is that it is possible to devise a graph where both $E(S,\bar S)$ and $E(T,\bar T)$ are large whereas $E(S,T)$ is fairly small. When $E(S,\bar S)$ and $E(T,\bar T)$ are big, the \emph{multiplicative error} $\eta$ given to both quantities might add too much noise to an estimation of $E(S,T)$. 

The problem relates to the dot-product estimation of the JL transform. It is a classical result that if $M$ is a distance-preserving matrix and $u$ and $v$ are two vectors s.t. $\|M(u+v)\|^2 \approx \|u+v\|^2$ and $\|M(u-v)\|^2 \approx \|u-v\|^2$ then it is possible to bound the difference $\left|Mu\cdot Mv - u\cdot v \right|$. But this bound is a function of $\|u\|$ and $\|v\|$, which in our case translates to a bound that depends on $\|E_G \1_S\|$ and $\|E_G \1_T\|$, both vectors of potentially large norms. 

\paragraph{Other Versions of JL.} The analysis in this works deals with the most basic JL transform, using normal Gaussians. We believe that qualitatively the same results should apply for other versions of the JL transform (e.g., with entries taken in $U_{[-1,1]}$). However, we are not certain whether the same results hold for \emph{sparse} transforms (see~\cite{DasguptaKS10}).

\paragraph{Low rank approximation of a given matrix.} The work of~\cite{HR11} gives a differentially private algorithm that outputs a low-rank approximation of a given matrix $A$, while adding additive error $> \min\{\sqrt d, \sqrt n\}$. Our work, which introduces much smaller noise (independent of $n$ and $d$), does not have such guarantees. Our algorithm could potentially be integrated into theirs. In particular, their algorithm is composed of two stages, and our technique greatly improves the first of the two. The crux of the second stage lies in devising a way to preserve differential privacy when multiplying a given (non-private) $X$ with a private database $A$ without introducing too large of an additive noise. Matrix multiplication via random projections might be such a way.

\ifx\ConferenceVersion\undefined 
\paragraph{Integration with the Multiplicative Weights mechanism.}  When the interactive Multiplicative Weights mechanism is given a user's query, it considers two possible alternatives: answering according to a synthetic database, or answering according to the Laplace mechanism. It chooses the latter alternative only when the two answers are far apart. Its utility guarantees rely on applying the Laplace mechanism only a bounded number of times. An interesting approach might be to add a third alternative, of answering according to the perturbed Laplacian we output. Hopefully, if most updates can be ``charged'' to answers provided by the perturbed Laplacian, it will allow us to improve privacy parameters (noise dependency on $n$).
\else 
\fi

\bibliographystyle{alpha}
\bibliography{JLPrivacy}

\appendix

\ifx\ConferenceVersion\undefined
\else
\section{Laplacians: Proof of Utility and Comparison}
\label{sec:apx_laplacians_completions}
We give the missing utility proof from Section~\ref{sec:perturbed_laplacian}, then detail the comparison of Algorithm~\ref{alg:output_perturbed_laplacian} with existing algorithms.

\input{util_proof_laplacian}
\input{comparison}
\fi

\ifx\ConferenceVersion\undefined 
\else 
\section{Covariance Matrices: Missing Proofs and Comparison}
\label{sec:apx_covariance_completions}
Below are the missing proofs from Section~\ref{sec:covariance_matrix}, and the comparison between Algorithm~\ref{alg:perturbed_covariance} with existing algorithms.

\input{util_proof_var}
\input{privacy_proof_var}
\input{comparison_var}
\fi 

\section{Facts from Linear Algebra}
\label{sec:apx_linear_algebra}

Below we prove the various facts from linear algebra that were mentioned in the body of the paper. We add the proofs, yet we comment that they are not new. In fact, existing literature~\cite{HornJohnson,Bhatia07} have documented proofs of the general theorems from which our facts are derived. Throughout this section, we denote the $i$-th eigenvalue (resp. the $i$-th singular value) of a given matrix $M$ in a descending order, assuming all eigenvalues are real, as $ev_i(M)$ (resp. as $sv_i(M)$).

\paragraph{Proving Fact~\ref{fact:Weyl_inequality}.} The fact uses the max-min characterization of the singular values of a matrix.

\begin{theorem}[Courant-Fischer Min-Max Principle]
For every matrix $A$ and every $1\leq i \leq n$, the $i$-th singular value of $A$ satisfies:
\[sv_i(A) = \max_{S: \dim(S) = i}\ \ \min_{x\in S:~ \|x\|=1} \ \langle A x, x\rangle \]
\end{theorem}

\begin{claim}[Weyl Inequality]
\label{clm:apx_weyl}
Let $A$ and $B$ be positive semidefinite matrices s.t. the matrix $E = B-A$ satisfies $x\T E x \geq 0$ for every $x$. Then for every $1\leq i \leq n$ it holds that
\[sv_i(A) \leq sv_i(B)\]
\end{claim}
\begin{proof}
Let $S_A$ be the $i$-dimensional subspace s.t. $sv_i(A) = \min_{x\in S_A:~ \|x\|=1} \langle Ax,x\rangle$. For every $x\in S_A$ we have that 
\[\langle Ax,x\rangle =\langle Ax,x\rangle +0 \leq \langle Ax,x\rangle + \langle Ex,x\rangle = \langle Bx,x\rangle\] so $sv_i(A) \leq \min_{x\in S_A:~ \|x\|=1} \langle Bx,x\rangle$. Thus $sv_i(B) = \max_S \min_{x\in S:~ \|x\|=1} \langle Bx,x\rangle \geq sv_i(A)$.
\end{proof}

Fact~\ref{fact:Weyl_inequality} is a direct application of Claim~\ref{clm:apx_weyl} to $L_G$ and $L_{G'}$.
\vspace{0.7cm}

\paragraph{Proving Fact~\ref{fact:inverse_order}.} The proof builds on the following two claims.

\begin{claim}
\label{clm:apx_inverse}
Let $A$ be a positive-semidefinite matrix. If $x\T A x \geq x\T x$ for every $x\in (Ker(A))^\perp$ then it also holds that $x\T A\mpinv x \leq x\T x$ for every $x\in (Ker(A))^\perp$.
\end{claim}
\begin{proof}
Denote the SVD of $A = V\Sigma^2 V\T = \sum_{i=1}^r \sigma_i^2 v_iv_i\T$, where $v_i$ is the $i$-th column of $V$. Fix $x \in (Ker(A))^\perp$ and observe that $x$ is span by the same $r$ vectors $\{v_1, v_2, \ldots, v_r\}$, so we can write $x = \sum_{i=1}^r \alpha_i v_i$. Denote $y = V\Sigma^{-1} V\T x = \sum_{i=1}^r \sigma_i^{-1} v_i v_i\T x$. We have that $y = \sum_{i=1}^r \alpha_i\sigma_i^{-1} v_i$ so $y\in (Ker(A))^\perp$. Therefore $y\T A y \geq y\T y$, but $y\T y = x\T A\mpinv x$ and $y\T A y = x\T x$.
\end{proof}

\begin{claim}
\label{clm:apx_two_PSDs}
Let $A$ and $B$ be two positive-semidefinite matrices s.t. $Ker(A) = Ker(B)$.  Then if for every $x$ we have that $x\T A x \leq x\T B x$ then $x\T A\mpinv x \geq x\T B\mpinv x$.
\end{claim}
\begin{proof}
We denote the SVD $A = V\Sigma^2 V\T$ and $B = W \Pi^2 W\T$. Because we can split any vector $x$ into the direct sum $x = x_0 + x_\perp$ where $x_0 \in Ker(A)=Ker(B)$ and $x_\perp \in (Ker(A))^\perp$, and since we have that the required inequality holds trivially for $x_0$, then we need to show it holds for $x_\perp$. Given any $z\in (Ker(A))^\perp$, set $y = V\Sigma^{-1} V\T z$. We know that $y\T A y \leq y\T B y$, and therefore
\[z\T z = y\T A y \leq y\T B y = z\T \left(V \Sigma^{-1} V\T W \Pi^2 W\T V \Sigma^{-1} V\T\right) z \stackrel{\mathrm {def}} = z\T C z\]
The above proves that $C$ is a positive semidefinite matrix whose kernel is exactly $Ker(A) = Ker(B)$, and so it follows from Claim~\ref{clm:apx_inverse} that $z\T z \geq z\T C\mpinv z$. Let $I|_{Ker(C)\perp}$ be the matrix which nullifies every element in $Ker(C)$, yet operates like the identity on $(Ker(C))^\perp$. One can easily check that $C\mpinv = V \Sigma V\T W \Pi^{-2} W\T V \Sigma V\T$ by verifying that indeed $C\mpinv C = C C\mpinv = I|_{Ker(C)\perp}$. So now, given $x$ we denote $z = V \Sigma^{-1} V\T x$ and apply the above to deduce $x\T B\mpinv x = z\T C\mpinv z \leq z\T z = x\T A\mpinv x$.
\end{proof}

Fact~\ref{fact:inverse_order} is a direct application of Claim~\ref{clm:apx_two_PSDs} to $L_G$ and $L_{G'}$.
\vspace{0.7cm}

\paragraph{Proving Fact~\ref{fact:Linskii_thm}.} Much like Claim~\ref{clm:apx_weyl} follows from the Courant-Fischer Min-Max principle, Lindskii's theorem follows from a generalization of this principle.

\begin{theorem}[Wielandt's Min-Max Principle.]
Let $A$ be a $n\times n$ symmetric matrix. Then for every $k$ and every $k$ indices $1\leq i_1 < i_2 < \ldots < i_k \leq n$ we have that
\[ \sum_{j=1}^k ev_{i_j}(A) = \max_{\substack{S_1 \subset S_2 \subset\ldots\subset S_k \\ \dim(S_j) = i_j}} \ \ \min_{\substack{x_j \in S_j: \\ x_j \textrm{ orthonormal}}} \ \sum_{j=1}^k \langle Ax_j, x_j\rangle \]
\end{theorem}

\begin{claim}[Linskii's theorem.]
\label{clm:apx_Linskii}
Let $A$ and $B$ be a $n\times n$ symmetric matrix. Denote $E = B-A$. Then for every $k$ and every $k$ indices $1\leq i_1 < i_2 < \ldots < i_k \leq n$ we have that
\[ \sum_{j=1}^k ev_{i_j}(B) \leq \sum_{j=1}^k ev_{i_j}(A) + \sum_{i=1}^k ev_{i}(E)\]
\end{claim}
\begin{proof}
Fix $i_1 < i_2 <\ldots  < i_k$ and let $T_1, T_2, \ldots, T_k$ the subspaces for which
\[\sum_{j=1}^k ev_{i_j}(B) = \min_{\substack{x_j \in T_j: \\ x_j \textrm{ orthonormal}}} \sum_{j=1}^k \langle Bx_j, x_j\rangle\] For every $v_1, v_2, \ldots, v_k$ orthonormal we have that $\sum_{j=1}^k \langle Bv_j, v_j\rangle = \sum_{j=1}^k \langle Av_j, v_j\rangle + \sum_{j=1}^k \langle Ev_j, v_j\rangle \leq \sum_{j=1}^k \langle Av_j, v_j\rangle + \sum_{i=1}^k ev_i(E)$, so \[\sum_{j=1}^k ev_{i_j}(B)\leq \min_{\substack{x_j \in T_j: \\ x_j \textrm{ orthonormal}}} \sum_{j=1}^k \langle Ax_j, x_j\rangle + \sum_{i=1}^k ev_i(E)\] and clearly
\[\sum_{j=1}^k ev_{i_j}(A) = \max_{\substack{S_1 \subset S_2 \subset\ldots\subset S_k \\ \dim(S_j) = i_j}} \min_{\substack{x_j \in S_j: \\ x_j \textrm{ orthonormal}}} \sum_{j=1}^k \langle Ax_j, x_j\rangle \geq \min_{\substack{x_j \in T_j: \\ x_j \textrm{ orthonormal}}} \sum_{j=1}^k \langle Ax_j, x_j\rangle\qedhere \]
\end{proof}
Now, Fact~\ref{fact:Linskii_thm} follows from Claim~\ref{clm:apx_Linskii}, and from the following observation of Weilandt.\footnote{We thank Moritz Hardt for bringing this observation to our attention.} Given a $m\times n$ matrix $M$, the matrix $N = \begin{pmatrix}0& M \cr M\T & 0\end{pmatrix}$ is symmetric and has eigenvalues which are (in descending order) $\big\{sv_1(A), sv_2(A), \ldots, sv_m(A), 0, 0, \ldots, 0, -sv_m(A), -sv_{m-1}(A), \ldots, -sv_1(A)\big\}$.
\end{document}

%% file: JLPrivacy_related.tex
Differential privacy was developed through a series of papers~\cite{DinurN03, Dwork06calibratingnoise, Chawla05towardprivacy, Blum2005}. Dwork et al~\cite{Dwork06calibratingnoise} gave the first formal definition and the description of the basic Laplace mechanism. Its Gaussian equivalent was defined in~\cite{DworkKMMN06}. Other mechanisms for preserving differential privacy include the Exponential Mechanism of McSherry and Talwar~\cite{mcsherry2007mechanism, blum2008learning}; the recent Multiplicative Weights mechanism of Hardt and Rothblum~\cite{hardt2010multiplicative} and its various extensions~\cite{HardtLM10, GuptaHRU11, GuptaRU12}; the Median Mechanism~\cite{roth2010interactive} and a boosting mechanism of Dwork et al~\cite{DRV10}. In addition, the classical Randomized Response (see~\cite{Warner65}) preserves differential privacy as discussed in recent surveys~\cite{dwork2010differential, Dwork11survey}. The task of preserving differential privacy when the given database is a graph or a social network was studied by Hay et al~\cite{HayLMJ09} who presented a privacy preserving algorithm for publishing the degree distribution in a graph. They also introduce multiple notions of neighboring graphs, one of which is for the change of a single edge. Nissim et al~\cite{nissim2007smooth} (see full version) studied the case of estimating the number of triangles in a graph, and Karwa et al~\cite{KRSY11} extended this result to other graph structures. Gupta et al~\cite{GuptaRU12} studied the case of answering $(S,T)$-cut queries, for two disjoint subsets of nodes $S$ and $T$. All latter works use the same notion of neighboring graphs as we do. In differential privacy it is common to think of a database as a matrix, but seldom one gives utility guarantees for queries regarding global properties of the input matrix. Blum et al~\cite{Blum2005} approximate the input matrix with the PCA construction by adding $O(d^2)$ noise to the input. The work of McSherry and Mironov~\cite{McSherryM09} (inspired by the Netflix prize competition) defines neighboring databases as a change in a single entry, and introduces $O(k^2)$ noise while outputting a rank-$k$ approximation of the input. The work of Hardt and Roth~\cite{HR11} gives a low-rank approximation of a given input matrix while adding $\min\{\sqrt d,\sqrt n\}$ noise by following the elegant framework of Halko et al~\cite{HalkoMT11}. According to~\cite{HR11}, a recent and not-yet-published work of Kapralov, McSherry and Talwar preserves rank-$1$ approximations of a given PSD matrix with error $O(n)$.

The body of work on the JL transform is by now so extensive that only a book may survey it properly~\cite{vempala2005random}. 
In the context of differential privacy, the JL lemma has been used to reduce dimensionality of an input prior to adding noise or other forms of privacy preservation. Blum et al~\cite{blum2008learning} gave an algorithm that outputs a sanitized dataset for learning large-margin classifiers by appealing to JL related results of~\cite{BalcanBV06}. Hardt and Roth~\cite{HR11} gave a privacy preserving version of an algorithm of~\cite{HalkoMT11} that uses randomize projections onto the image space of a given matrix. Blum and Roth~\cite{BlumRoth11} used it to reduce the noise added to answering sparse queries. 
The way the JL lemma was applied in these works is very different than the way we use it.

%% file: util_proof_laplacian.tex
\begin{proof}[Proof of Theorem~\ref{thm:approx_cut}]
Let us denote $G$ as the input graph for Algorithm~\ref{alg:output_perturbed_laplacian}, and $H$ as the graph resulting from the changes in edge-weights Algorithm~\ref{alg:output_perturbed_laplacian} makes. Therefore, \[L_H = L_{\frac w n K_n} + L_{(1-\frac w n)G} = \frac w n L_{K_n} +\left(1-\frac w n\right) L_G \]
Fix $S$. The JL Lemma (Theorem~\ref{thm:JL}) assures us that w.p. $\geq 1-\nu$ we have
\[ (1-\eta) \1_S\T L_H \1_S \leq \1_S\T \tilde L \1_S \leq (1+\eta) \1_S\T L_H \1_S\]
The proof now follows from basic arithmetic and the value of $w$.
\begin{eqnarray*}
& R(S) & \leq \frac 1 {1-\frac w n} \left((1+\eta)\1_S\T L_H \1_S - w \frac {s(n-s)}n\right) \cr 
&& = \frac 1 {1-\frac w n} \left((1+\eta)\frac w n s(n-s) + (1+\eta)(1-\frac w n)\1_S\T L_G \1_S - w \frac {s(n-s)}n\right) \cr
&& \leq (1+\eta) \Phi_G(S) + \frac 1 {1-\frac w n} \eta w \cdot s = (1+\eta) \Phi_G(S) + \tau 
\end{eqnarray*}
where $\tau \leq 2\eta w\cdot s$. The lower bound is obtained exactly the same way.
\end{proof}

%% file: comparison.tex
\subsection{Discussion and Comparison with Other Algorithms}
\label{subsec:comparison_of_algorithsm}

Recently, Gupta et al~\cite{GuptaRU12} have also considered the
problem of answering cut-queries while preserving differential
privacy, examining both an iterative database construction approach
(e.g., based on the multiplicative-weights method) and a 
randomized-response approach.  Here, we compare this and other methods
to our algorithm. 
We compare them along several axes: the dependence on $n$ and $s$
(number of vertices in $G$ and in $S$ resp.), the dependence on
$\epsilon$, and the dependence on $k$ -- the number of
queries answered by the mechanism.  Other parameters are omitted. The bottom line is that for a long
non-adaptive query sequence, our
approach dominates in the case that $s = o(n)$. 
The results are summarized in Table~\ref{tab:mechanisms_comparison}. 

Note, comparing the dependence on $k$ for interactive and non-interactive mechanisms is not straight-forward. In general, non-interactive mechanisms are more desirable than interactive mechanisms, because interactive mechanisms require a central authority that serves as the only way users can interact with the database. However, interactive mechanisms can answer $k$ \emph{adaptively chosen} queries. In order for non-interactive mechanisms to do so, they have to answer correctly on $\min\{\exp(O(k)), 2^n\}$ queries. This is why outputting a sanitized database is often considered a harder task than interactively answering user queries. We therefore compare answering $k$ adaptively chosen queries for interactive mechanisms, and $k$ \emph{predetermined} queries for non-interactive mechanism.

\subsubsection{Our Algorithm}
\label{subsubsec:discussion_our_algorithm}

Clearly, our algorithm is non-interactive. As such, if we wish to
answer correctly w.h.p. a set of $k$ \emph{predetermined} queries, we
set $\nu' = \nu/k$, and deduce that the amount of noise added to each
query is $O(s \sqrt{\log(k)}/\epsilon)$. So, if we wish to
answer all $2^n$ cut queries correctly, our noise is set to $\tilde O(s \sqrt{n} /\epsilon)$. An interesting observation is that in such a case we aim to answer all $2^n$ queries, we generate a iid normal matrix of size $r\times n$ where $r>n$. Therefore, we now apply the JL transform to \emph{increase} the dimensionality of the problem rather than decreasing it. This clearly sets privacy preserving apart from all other applications of the JL transform.

In addition, we comment that our algorithm can be implemented in a \emph{distributed} fashion, where node $i$ repeats the following procedure $r$ times (where $r$ is the number of rows in the matrix picked by Algorithm~\ref{alg:output_perturbed_laplacian}): First, $i$ picks $n-i-1$ iid samples from $\mathcal{N}(0,1)$ and sends the $j$-th sample, $x_j$, to node $i+j$. Once node $i$ receives $i-1$ values from nodes $1,2,\ldots, i-1$, it outputs the weighted sum $\sum_{j \neq i} (-1)^{\{j< i\}}x_{j} \left(\sqrt{\frac w n} + w_{i,j}(1-\sqrt{\frac w n})\right)$ (where $(-1)^{\{j< i\}}$ denotes $-1$ if $j<i$, or $1$ otherwise). 

\subsubsection{Na\"ively Adding Laplace Noise}
\label{subsubsec:lapalacian}

The most basic of all differentially private mechanisms is the classical Laplace mechanism
which is interactive. A user poses a cut-query $S$ and the mechanism replies with $\Phi_G(S) + Lap(0,\epsilon^{-1})$ (since the global sensitivity of cut-queries is $1$). The composition theorem of~\cite{DRV10} assures us that for $k$ queries we preserve $(O(\sqrt{k} \epsilon),\delta)$-privacy. As a result, the mechanism completely obfuscates the true answer if $k \geq n^4$ and even for $k=n^2$ has noise proportional to $n / \epsilon$.

\subsubsection{The Randomized Response Mechanism}
\label{subsubsec:RR}

The ``Randomized Response'' algorithm perturbs the edges of a graph in a way  that allows us to publish the result and still preserve privacy. Given $G$, the Randomized Response algorithm constructs a weighted graph $H$ where for every $u,v\in V(G)$, the weight of the edge $(u,v)$ in $H$, denoted $w_{u,v}'$, is chosen independently to be either $1$ or $-1$. Each edge picks its weight independently, s.t. $\Pr[w'_{u,v} = 1] = \frac {1 + \epsilon w_{u,v}} 2$ and $\Pr[w'_{u,v} = -1] = \frac {1 - \epsilon w_{u,v}} 2$. 
Clearly, this algorithm maintains $\epsilon$-differential edge privacy: two neighboring graphs differ on a single edge, $(a,b)$, and obviously \[\Pr[w'_{a,b} =1 \ | \ w_{a,b} = 1] \leq (1+\epsilon)\Pr[w'_{a,b} =1 \ | \ w_{a,b} = 0]\]
In addition, it is also evident that for every nonempty $S\subsetneq V(G)$, we have that $\E[\sum_{u\in S, v\in \bar S} w'_{u,v}] = \epsilon \sum_{u\in S, v\notin S} w_{u,v} = \epsilon \Phi_G(S)$, yet the variance of this r.v. is $\Omega(s(n-s))$. Therefore, a classical Hoeffding-type bound gives that for any nonempty $S\subsetneq V(G)$ we have that for every $0<\nu<1/2$, \[\Pr\left[\ \left|\frac 1 \epsilon \sum_{u\in S, v\in \bar S} w'_{u,v} - \Phi_G(S) \right| > 
\frac {\sqrt{2\log(1/\nu) s(n-s)}} \epsilon\ \right] \leq 2\nu\]

Observe that while $\sqrt{s(n-s)}$ is a comparable with $s$ when $s =
\Omega(n)$, there are cuts (namely, cuts with $s = O(1)$) where
$\sqrt{\frac {n-s}{s}} = \Omega(\sqrt{n})$.   More generally,
the additive noise of Randomized Response is a factor $\sqrt{n/s}$ worse than our
algorithm. We comment that the Randomized Response
algorithm can also be performed in a distributed fashion, and in
contrast to our algorithm, it has no multiplicative error. In addition,
the above analysis holds for any linear combination of edge, not just
the $s(n-s)$ potential edges that cross the $(S,\bar S)$ cut. So given
$E' \subset E(G)$ it is possible to approximate $\sum_{e\in E'} w_e$
up to $\pm \frac {\sqrt{|E'|\log(1/\nu)}} \epsilon$ w.p. $\geq
1-2\nu$. In particular, for queries regarding an $(S,T)$-cut (where $S,T$ are two disjoint subsets of vertices) we can estimate the error up to $\pm \frac {\sqrt{|S||T|\log(1/\nu)}} \epsilon$. We also comment that the version of Randomized Response presented here differs slightly from the version of~\cite{GuptaRU12}. In particular, it is possible to address their concern regarding outputting a sanitized graph with non-negative weights by an affine transformation taking $\{-1,1\} \to \{0,1\}$.

\subsubsection{Exponential Mechanism / BLR}
\label{subsubsec:BLR}

The exponential mechanism~\cite{mcsherry2007mechanism, blum2008learning} is a non-interactive privacy preserving mechanism, which is typically intractable. To implement it for cut-queries one needs to (a) specify a range of potential outputs and (b) give a scoring function over potential outputs s.t. a good output's score is much higher than all bad outputs' scores.

One such set of potential outputs is derived from edge-sparsifiers. Given a graph $G$ we say that $H$ is an edge-sparsifier for $G$ if for any nonempty $S\subsetneq V(G)$ it holds that $\Phi_H(S) \in (1\pm\eta) \Phi_G(S)$. There's a rich literature on sparsifiers (see~\cite{BenczurK96, SpielmanT04, SpielmanS08}), and the current best known construction~\cite{BatsonSS09} gives a (weighted) sparsifier with $O(n/\eta^2)$ edges with all edge-weights $\leq \poly(n)$. By describing every edge's two endpoints and weight, we have that such edge-sprasifiers can be described using $O(n\log(n))$ bits (omitting dependence on $\eta$). Thus, the set of all sparsifiers is bounded above by $\exp(O(n\log(n)))$. Given an input graph $G$ and a weighted graph $H$, we can score $H$ using $q(G,H) = \max_S \left\{\min_{\alpha:\ |\alpha - 1| \leq \eta} \left| \Phi_H(S)/\alpha - \Phi_G(S)\right|\right\}$. Observe that if we change $G$ to a neighboring graph $G'$, then the score changes by at most $1$.

Putting it all together, we have that given input $G$ the exponential mechanism gives a score of $e^{-\epsilon q(G,H)/2}$ to each possible output. The edge-sparsifier of $G$ gets score of $1$, whereas every graph with $q(G,H) > \tau$ gets a score of $e^{-\epsilon \tau/2}$. So if we wish to claim we output a graph whose error is $>\tau$ w.p. at most $\nu$, then we need to set $\exp(n\log(n) - \epsilon\tau/2) \leq \nu$. It follows that $\tau$ is proportional to $n\log(n)/\epsilon$. Note however that the additive error of this mechanism is independent of the number of queries it answers correctly.

We comment that even though we managed to find a range of size $2^{O(n\log(n))}$, it is possible to show that the range of the mechanism has to be $2^{\Omega(n)}$. (Fix $\alpha < 1/2$ and think of a set of inputs $\mathcal{G}$ where each $G\in\mathcal{G}$ has $n/2$ vertices with degree $n^\alpha$ and $n/2$ vertices with degree $n^{2\alpha}$. Preserving all cuts of size $1$ up to $(1\pm\eta)$ requires our output to have vertices of degree $> (1-\eta)n^{2\alpha}$ and vertices of degree $<(1+\eta)n^{\alpha}$. Therefore, by representing vertices of high- and low-degree using a binary vector, there exists an injective mapping of balanced $\{0,1\}^n$-vectors onto the set of potential outputs.) Thus, unless one can devise a scoring function of lower sensitivity, the exponential mechanism is bounded to have additive error proportional to $n/\epsilon$.

\subsubsection{The Multiplicative Weights Mechanism}
\label{subsubsec:weighted_maj}

The very elegant Multiplicative Weights mechanism of Hardt and Rothblum~\cite{hardt2010multiplicative} can be adapted as well for answering cut queries. In the Multiplicative Weights mechanism, a database is represented by a histogram over all $N$ ``types'' of individuals that exist in a certain universe. In our case, each pair of vertices is a type, and each entry in the database is an edge detailing its weight. Thus, $N = \binom{n}{2}$ and the database length $=|E|$,\footnote{Observe that it is not possible to assume $|E| = O(n)$ using sparsifiers, because sparsifiers output a \emph{weighted} graph with edge-weights $O(n)$. Since the Multiplicative Weights mechanism views the database as a histogram the overall resolution of the problem remains roughly $n^2$ in the worst case.}  and each query $S$ corresponds to taking a dot-product between this histogram the ${\binom n 2}$-length binary vector indicating the edges that cross the cut. Plugging these parameters into the main theorem of~\cite{hardt2010multiplicative}, we get an adaptive mechanism that answers $k$ queries with additive noise of $\tilde O(\sqrt{|E|} \log(k) /\epsilon)$.

We should mention that the Multiplicative Weights mechanism, in contrast to ours, always answers correctly with no multiplicative error and can deal with $k$ adaptively chosen queries. Furthermore, it allows one to answer any linear query on the edges, not just cut-queries and in particular answer $(S,T)$-cut queries. However, its additive error is bigger than ours, and should we choose to set $k = 2^n$ (meaning, answering all cut-queries) then its additive error becomes $\tilde O(n\sqrt{|E|}/\epsilon)$ (in contrast to our $O(s\sqrt{n}/\epsilon)$).

Gupta et al~\cite{GuptaRU12} have improved on the bounds on the Multiplicative Weights mechanism by generalizing it as a ``Iterative Database Construction'' mechanism, and providing a tighter analysis of it. In particular, they have reduced the dependency on $\epsilon$ to $1/\sqrt\epsilon$. Overall, their additive error is $\tilde O(\sqrt{|E|\log(k)}/\sqrt\epsilon)$, which for the case of all cut-queries is $\tilde O(\sqrt{n|E|/\epsilon})$.

\begin{table}[t]
\footnotesize
\centering
\begin{tabular}{| p{2.4cm} |c|c|p{1.0cm}|p{1.0cm}|p{1.0cm}|p{2.5cm}|}
\hline
\centering Method & \parbox{3cm}{\centering Additive Error\\for any $k$} & \parbox{3cm}{\centering Additive Error\\ for all Cuts} & {Multi-plicative Error?} & {Inter-active?} & {Tract-able?} & Comments \cr
\hline
Laplace Noise \cite{Dwork06calibratingnoise}& $O(\sqrt{k}/\epsilon)$ & $O(2^{n/2}\epsilon)$ & \centering\XSolid & \centering\Checkmark & \centering\Checkmark & \cr
\hline
Randomized Response & $O(\sqrt{sn\log(k)}/\epsilon)$ & $O(n\sqrt{s}/\epsilon)$ &\centering\XSolid & \centering\XSolid & \centering\Checkmark & Can be distributed; answers $(S,T)$-cut queries \cr 
\hline
Exponential Mechanism \cite{mcsherry2007mechanism, blum2008learning} & $O(n\log(n)/\epsilon)$ & $O(n\log(n)/\epsilon)$ & \centering\Checkmark &\centering \XSolid & \centering\XSolid  & Error ind. of $k$ \cr
\hline
\parbox{2cm}{MW~\cite{hardt2010multiplicative} \\ IDC~\cite{GuptaRU12}} &  \parbox{3.0cm}{\centering $\tilde O(\sqrt{|E|}\log(k)/\epsilon)$ \\ $\tilde O(\sqrt{|E|\log(k)/ \epsilon})$ } & 
\centering \parbox{2cm}{\centering $\tilde O(n\sqrt{|E|}/\epsilon)$ \\ $\tilde O(\sqrt{n|E|/ \epsilon})$} &\centering{\XSolid} & \centering{\Checkmark} &\centering{\Checkmark} & {Answers $(S,T)$-cut queries} \cr 
\hline
\hline
JL & $O(s\sqrt{\log(k)} /\epsilon)$ & $\tilde O(s\sqrt{n})/\epsilon)$ & \centering\Checkmark & \centering\XSolid & \centering\Checkmark & Can be distributed\cr 
\hline
\end{tabular}
\caption{Comparison between mechanisms for answering cut-queries. $\epsilon$ -- privacy parameter; $n$ and $|E|$  -- number of vertices and edges resp.; $s$ -- number of vertices in a query; $k$ -- number of queries.\label{tab:mechanisms_comparison}}
\end{table}

%% file: util_proof_var.tex
\begin{proof}[Proof of Theorem~\ref{thm:utility_directional_variance}]
Again, the proof is immediate from the JL Lemma, and straight-forward arithmetics give that for every $x$ w.p. $\geq 1-\nu$ we have that
\[(1-\eta) \Phi_A(x) - \eta w^2 \leq R(x) \leq (1+\eta)\Phi_A(x) + \eta w^2\] so $\tau = \eta w^2$.
\end{proof}

%% file: privacy_proof_var.tex
\begin{proof}[Proof of Theorem~\ref{thm:perturbed_covariance_preserves_privacy}]
Fix two neighboring $A$ and $A'$. We often refer to the gap matrix $A'-A$ as $E$. Observe, $E$ is a rank-$1$ matrix, which we denote as the outer-product $E = e_i v\T$ ($e_i$ is the indicator vector of row $i$ and $v$ is a vector of norm $1$). As such, the singular values of $E$ are exactly $\{1,0,\ldots, 0\}$.\footnote{For convenience, we ignore the part of the algorithm that subtracts the mean of the rows of $A$. Observe that if $E = A-A'$ then after subtracting the mean from each row, the difference between the two matrices is $\tilde{e_i}\T v$ where $\tilde{e_i}$ is simply subtracting $1/n$ from each coordinate of $e_i$. Since $\|\tilde{e_i}\| < \|e_i\|$, this has no effect on the analysis.}

The proof of the theorem is composed of two stages. The first stage is the simpler one. We ignore step $4$ of Algorithm~\ref{alg:perturbed_covariance} (shifting the singular values), and work under the premise that both $A$ and $A'$ have singular values no less than $w$. In the second stage we denote $B$ and $B'$ as the results of applying step $4$ to $A$ and $A'$ resp., and show what adaptations are needed to make the proof follow through. 

\paragraph{Stage $1$.}

\noindent We assume step $4$ was not applied, and all singular values of $A$ and $A'$ are at least $w$.
 
As in the proof of Theorem~\ref{thm:perturbed_laplacian_preserves_privacy}, the proof follows from the assumption that Algorithm~\ref{alg:perturbed_covariance} outputs $O\T = A\T M$ (which clearly allows us to reconstruct $\tilde C = \frac 1 r O\T O$). Again $O\T$ is composed of $r$ columns each is an iid sample from $A\T Y$ where $Y\sim \mathcal{N}(0, I_{n\times n})$. We now give the analogous claim to Claim~\ref{clm:inequalities_of_privacy}.
\begin{claim}
\label{clm:inequalities_of_privacy_for_covariance}
Fix $\epsilon_0 = \frac {\epsilon} {\sqrt{4r\ln(2/\delta)}}$ and $\delta_0 = \frac \delta {2r}$. Denote $S = \{x:\ e^{-\epsilon_0}{\PDF_{{A'}\T Y}(x)} \leq PDF_{A\T Y}(x) \leq e^{\epsilon_0}\PDF_{{A'}\T Y}(x)\}$. Then $\Pr[S]\geq 1-\delta_0$.
\end{claim}
Again, the composition theorem of~\cite{DRV10} along with the choice of $r$ gives that overall we preserve $(\epsilon,\delta)$-differential privacy.
\end{proof}

\begin{proof}[Proof of Claim~\ref{clm:inequalities_of_privacy_for_covariance}]
The proof mimics the proof of Claim~\ref{clm:inequalities_of_privacy}, but there are two subtle differences. First, the problem is simpler notation-wise, because $A$ and $A'$ both have full rank due to Algorithm~\ref{alg:perturbed_covariance}. Secondly, the problem becomes more complicated and requires we use some heavier machinery, because the singular values of $A'$ aren't necessarily bigger than the singular values of $A$. Details follow.

First, let us formally define the $\PDF$ of the two distributions. Again, we apply the fact that $A\T Y$ and ${A'}\T Y$ are linear transformations of $\mathcal{N}(0, I_{n\times n})$.
\begin{eqnarray*}
&&\PDF_{A\T Y}(x) = \frac 1 {\sqrt{(2\pi)^{d} \det(A\T A)}} \exp(-\frac 12 x\T (A\T A)^{-1} x) \cr
&&\PDF_{{A'}\T Y}(x) =  \frac 1 {\sqrt{(2\pi)^{d} \det({A'}\T A')}} \exp(-\frac 12 x\T ({A'}\T A')^{-1} x)\cr
\end{eqnarray*}

Our proof proceeds as follows. First, we show
\begin{equation}
\label{eq:det_ratio}
e^{-\epsilon_0/2} \leq \sqrt{\frac{\det({A'}\T{A'})}{\det(A\T A)}} \leq e^{\epsilon_0 / 2}
\end{equation}
Then we show that no matter whether we sample $x$ from $A\T Y$ or from ${A'}\T Y$, we have that
\begin{equation}
\label{eq:exp_ratio}
\Pr_x\left[ \frac 1 2 \left| x\T \left( (A\T A)^{-1} - ({A'}\T {A'})^{-1}\right) x\right| \geq \epsilon_0/2 \right] \leq \delta_0
\end{equation}
Clearly, combining both~\eqref{eq:det_ratio} and~\eqref{eq:exp_ratio} proves the claim.

Let us prove~\eqref{eq:det_ratio}. Denote the SVD of $A = U\Sigma V\T$ and $A' = U' \Lambda {V'}\T$, where the singular values of $A$ are $\sigma_1 \geq \sigma_2 \geq \ldots \geq \sigma_d > 0$ and the singular values of $A'$ are $\lambda_1 \geq \lambda_2 \geq \ldots \geq \lambda_d > 0$. Therefore we have $A\T A = V \Sigma^2 V\T$, ${A'}\T A' = V' \Lambda^2 {V'}\T$ and also $(A\T A)^{-1} = V \Sigma^{-2} V\T$, $({A'}\T A')^{-1} = V' \Lambda^{-2} {V'}\T$. Thus $\det(A\T A) = \prod_{i=1}^d \sigma_i^2$ and $\det({A'}\T A') = \prod_{i=1}^d \lambda_i^2$. 

This time, in order to bound the gap $\sum_i (\lambda_i^2 - \sigma_i^2)/\sigma_i^2$ it isn't sufficient to use the trace of the matrices. Instead, we invoke an application of Lindskii's theorem (Theorem 9.4 in~\cite{Bhatia07}).
\begin{fact}[Linskii]
\label{fact:Linskii_thm}
For every $k$ and every $1\leq i_1 < i_2 < \ldots < i_k\leq n$ we have that \[\sum_{j=1}^k \lambda_{i_j} \leq \sum_{j=1}^k \sigma_{i_j} + \sum_{i=1}^k \textrm{sv}_i(E)\] where $\{\textrm{sv}_i(E)\}_{i=1}^n$ are the singular values of $E$ sorted in a descending order.
\end{fact}
As a corollary, because $E$ has only $1$ non-zero singular value, we denote $Big = \{i : \ \lambda_i > \sigma_i \}$ and deduce that $\sum_{i\in Big} \lambda_i - \sigma_i \leq 1$. Similarly, since the singular values of $E$ and of $(-E)$ are the same, we have that $\sum_{i\notin Big} \sigma_i - \lambda_i \leq 1$. Using this, proving~\eqref{eq:det_ratio} is straight-forward:
\begin{eqnarray*}
& \sqrt{\prod_i \frac {\lambda_i^2} {\sigma_i^2}} & \leq \prod_{i\in Big} \left(1 + \frac {\lambda_i - \sigma_i} {\sigma_i}\right) \leq \exp\left(\frac 1 { w}\sum_{i\in Big} \lambda_i - \sigma_i\right) \leq e^{w^{-1}} \leq e^{\epsilon_0/2} 
\end{eqnarray*}
and similarly, $\sqrt{\prod_i \frac {\sigma_i^2}{\lambda_i^2} }\leq e^{\epsilon_0/2}$.

We turn to proving~\eqref{eq:exp_ratio}. We start with the following derivation.
\begin{eqnarray*}
& x\T (A\T A)^{-1} x - x\T ({A'}\T {A'})^{-1} x &= x\T (A\T A)^{-1} ({A'}\T {A'}) ({A'}\T {A'})^{-1} x - x\T ({A'}\T {A'})^{-1} x = \cr 
&& = x\T (A\T A)^{-1} ({(A+E)}\T {(A+E)}) ({A'}\T {A'})^{-1} x - x\T ({A'}\T {A'})^{-1} x \cr 
&& = x\T (A\T A)^{-1} (A\T E + E\T A') ({A'}\T {A'})^{-1} x 
\end{eqnarray*}
and using the SVD and denoting $E = e_i v\T$, we get
\begin{eqnarray*}
& x\T (A\T A)^{-1} x - x\T ({A'}\T {A'})^{-1} x &= x\T \left(V\Sigma^{-1} U\T \right)e_i \cdot v\T \left({V'}\Lambda^{-2} {V'}\T\right)x \cr 
&&+ x\T \left(V\Sigma^{-2} V\T \right)v \cdot e_i\T \left({U'}\Lambda^{-1} {V'}\T\right)x
\end{eqnarray*}
So now, assume $x$ is sampled from $A\T Y$. (The case of ${A'}\T Y$ is symmetric. In fact, the names $A$ and $A'$ are interchangeable.) That is, assume we've sampled $y$ from $Y \sim \mathcal{N}(0, I_{n\times n})$ and we have $x = A\T y = V \Sigma U\T y$ and equivalently $x = ({A'}\T - E\T) y = V' \Lambda {U'}\T y - v e_i\T y$. The above calculation shows that 
\[ \left| x\T (A\T A)^{-1} x - x\T ({A'}\T {A'})^{-1} x \right| \leq term_1 \cdot term_2 + term_3\cdot term_4 \]
where for $i=1,2,3,4$ we have $term_i = |vec_i \cdot y|$ and
\begin{align*}
&vec_1 = U\Sigma V\T V \Sigma^{-1} U e_i = e_i , 
&&\textrm{ so } \|vec_1\| = 1 & \\ 
&vec_2 = U' \Lambda^{-1} {V'}\T v - e_i v\T V' \Lambda^{-2} {V'}\T v,
&&\textrm{ so } \|vec_2\| \leq \frac 1 {\lambda_d} +\frac 1 {\lambda_d^2}  & \\
&vec_3 = U\Sigma^{-1} V\T v, 
&&\textrm{ so } \|vec_3\| \leq \frac 1 {\sigma_d} & \\
&vec_4 = e_i - e_i v\T {V'} \Lambda^{-1} {U'}\T e_i, 
&&\textrm{ so } \|vec_4\| \leq 1 + \frac 1 {\lambda_d} &\\
\end{align*}
Recall that all singular values, both of $A$ and $A'$, are greater than ${w}$ and that $vec_i \cdot y \sim \mathcal{N}(0, \|vec_i\|^2)$, so w.p. $\geq 1 - \delta_0$ we have that for every $i$ it holds that $term_i \leq \sqrt{\ln(4/\delta_0)} \|vec_i\|$ so
\[\left| x\T (A\T A)^{-1} x - x\T ({A'}\T {A'})^{-1} x \right| \leq 2(\frac 1 {w} + \frac 1 {w^2})\ln(4/\delta_0) \leq \frac {4\ln(4/\delta_0)} w \leq \epsilon_0 \] this concludes the proof in our first stage.

\paragraph{Stage $2$.}

\noindent We assume step $4$ was applied, and denote $B = U(\sqrt{\Sigma^2 + w^2I})V\T$ and $B' = U'(\sqrt{\Lambda^2 + w^2 I}){V'}\T$. We denote the singular values of $B$ and $B'$ as $\sigma_1^B \geq \sigma_2^B \geq \ldots \geq \sigma_d^B$ and $\lambda_1^B\geq \lambda_2^B \geq \ldots \geq \lambda_d^B$ resp. Observe that by definition, for every $i$ we have $(\sigma_i^B)^2 =\sigma_i^2 + w^2$ and $(\lambda_i^B)^2 = \lambda_i^2 + w^2$.

Again, we assume we output $O\T = B\T Y$, and compare $X = B\T Y$ to $X' = {B'}\T Y$. The theorem merely requires Claim~\ref{clm:inequalities_of_privacy_for_covariance} to hold, and they, in turn, depend on the following two conditions.
\begin{align}
& e^{-\epsilon_0/2} \leq \sqrt{\frac{\det({B'}\T{B'})}{\det(B\T B)}} \leq e^{\epsilon_0 / 2} &\label{eq:det_ratio_2} \\
 & \Pr_x\left[ \frac 1 2 \left| x\T \left( (B\T B)^{-1} - ({B'}\T {B'})^{-1}\right) x\right| \geq \epsilon_0/2 \right] \leq \delta_0 & \label{eq:exp_ratio_2}\\
\notag
\end{align}
The second stage deals with the problem that now, the gap $\Delta = B' - B$ is not necessarily a rank-$1$ matrix. However, what we show is that all stages in the proof of Claim~\ref{clm:inequalities_of_privacy_for_covariance} either rely on the singular values or can be written as the sum of a few rank-$1$ matrix multiplications. 

The easier part is to claim that Eq.~\eqref{eq:det_ratio_2} holds. The analysis is a simple variation on the proof of Eq.~\eqref{eq:det_ratio}. Fact~\ref{fact:Linskii_thm} still holds for the singular values of $A$ and $A'$. Observe that $\lambda_i^B > \sigma_i^B$ iff $\lambda_i > \sigma_i$. And so we have 
\[\sqrt{\prod_i \frac {(\lambda_i^B)^2}{(\sigma_i^B)^2}}  \leq \sqrt{\prod_{i\in Big} \frac{\lambda_i^2+w^2}{\sigma_i^2+w^2}} \leq \sqrt{\prod_{i\in Big} \frac{\lambda_i^2}{\sigma_i^2}}\] and the remainder of the proof follows.

We now turn to proving Eq.~\eqref{eq:exp_ratio_2}. We start with an observation regarding ${A'}\T A$ and ${B'}\T B'$.
\begin{align*}
{A'}\T A' & = (A+E)\T(A+E) = A\T A + {A'}\T E + E\T A \cr
{B}\T B\  & = V(\Sigma^2 + w^2 I)V\T = V\Sigma^2 V\T + w^2 I = A\T A + w^2I \cr
{B'}\T B' &=  V'(\Lambda^2+w^2I) {V'}\T = {A'}\T A' + w^2 I \cr
\Rightarrow \ {B'}\T B' - {B}\T B &= {A'}\T E + E\T A
\end{align*}
Now we can follow the same outline as in the proof of~\eqref{eq:exp_ratio}. Fix $x$, then:
\begin{eqnarray*}
& x\T (B\T B)^{-1} x - x\T ({B'}\T {B'})^{-1} x &= x\T (B\T B)^{-1} ({B'}\T {B'}) ({B'}\T {B'})^{-1} x - x\T ({B'}\T {B'})^{-1} x = \cr 
&& = x\T (B\T B)^{-1} \left[ B\T B + {A'}\T E + E\T A \right] ({B'}\T {B'})^{-1} x - x\T ({B'}\T {B'})^{-1} x \cr 
&& = x\T (B\T B)^{-1} \left[ {A'}\T E + E\T A \right] ({B'}\T {B'})^{-1} x \cr 
&& = x\T (B\T B)^{-1} (A\T+E\T)e_i \ \cdot\  v\T  ({B'}\T {B'})^{-1} x  \cr
&& + x\T (B\T B)^{-1} v \ \ \cdot\  e_i\T \left(A' - E\right) ({B'}\T {B'})^{-1} x \cr 
\end{eqnarray*}
It is straight-forward to see that the $i$-th spectral values of $(B\T B)^{-1} A$ is $\frac {\sigma_i} {\sigma_i^2 + w^2} \leq \frac 1 {\sqrt{\sigma_i^2 + w^2}} \leq 1/w$, and similarly for the spectral values of $({B'}\T B')^{-1} {A'}$. We now proceed as before and partition the above sum into multiplications of pairs of terms where $term_i \leq |vec_i \cdot y|$, and $y$ is sampled from $\mathcal{N}(0,I_{n\times n})$ and $x = B\T y$:
\begin{eqnarray*}
& x\T (B\T B)^{-1} x - x\T ({B'}\T {B'})^{-1} x &= 
y\T \left[ B (B\T B)^{-1} (A\T+E\T)e_i \right]\ \cdot\  \left[v\T  ({B'}\T {B'})^{-1} B\T \right] y  \cr
&& + y\T \left[ B (B\T B)^{-1} v \ \right]\ \cdot\ \left[ e_i\T \left(A' - E\right) ({B'}\T {B'})^{-1} B\T\right] y\cr 
\end{eqnarray*}

Lastly, we need to bound all terms that contain the multiplication $({B'}\T {B'})^{-1} B\T y$ in comparison to $({B'}\T {B'})^{-1} {B'}\T y = {B'}\mpinv y$. For instance, take the $term = |vec\T y|$ for $vec\T = e_i\T \left(A' - E\right) ({B'}\T {B'})^{-1} B\T$, and define it as $vec\T = z\T B\T$. We can only bound $\|Bz\|$ using $\sigma_1^B / (\lambda_d^B)^2$, whereas we can bound $\|B'z\|$ with $1/\lambda_d^B < 1/w$. In contrast to before, we do not use the fact that $B\T y = ({B'} - \Delta)\T y$. Instead, we make the following derivations.

First, we observe that for every vector $z$ we have that $\|B'z\| \geq \|A'z\|$ and $\|B'z\| \geq w \|z\|$. Using the fact that $B\T B - \ {B'}\T B'= -{A'}\T E - E\T A$, a simple derivation gives that $\|Bz\|^2 \leq \left(\|B'z\|+\|z\|\right)^2 \leq \left(1+\frac 1 w\right)^2\|B'z\|^2$, and vice-versa. So if $y$ is s.t. $\frac {|z\T B\T y|} {\left(1+\frac 1 w\right)\|B'z\|} > Threshold$ then $\frac {|z\T B\T y|} {\|Bz\|} > Threshold$.
Observe that $z\T B\T y$ is distributed like $\mathcal{N}(0, \|Bz\|^2) = \|Bz\|\mathcal{N}(0,1)$, and so we have that for every $\delta'>0$
\begin{align*}
 \Pr\left[ |z\T B\T y| \geq \sqrt{\log(1/\delta')} \left(1+\frac 1 w\right)\|B'z\| \right] & =  \Pr\left[ \left(\left(1+\frac 1 w\right)\|B'z\|\right)^{-1}|z\T B\T y| \geq \sqrt{\log(1/\delta')}  \right]
\cr 
& \leq \Pr\left[ \left(\|Bz\|\right)^{-1}|z\T B\T y| \geq \sqrt{\log(1/\delta')}  \right] \leq \delta' \qedhere
\end{align*}
\end{proof}

%% file: comparison_var.tex
\subsection{Comparison with Other Algorithms}
\label{subsec:comparison_variance}

To the best of our knowledge, no previous work has studied the problem of preserving the variance of $A$ in the same formulation as us. We deal with a scenario where users pose the directions on which they wish to find the variance of $A$. Other algorithms, that publish the PCA or a low-rank approximation of $A$ without compromising privacy (see Section~\ref{subsec:related}), provide users with specific directions and variances. These works are not comparable with our algorithm, as they give a different utility guarantee. For example, low-rank approximations aim at nullifying the projection of $A$ in certain directions. 

Here, we compare our method to the Laplace mechanism, the Multiplicative Weights mechanism and Randomized Response. The bottom line is clear: our method allows one to answer directional variance queries with additive noise which is independent of the given input. Other methods require we add random noise that depends on the size of the matrix, assuming we answer polynomially many queries. 

Our notation is as follows. $n$ denotes the number of rows in the matrix (number of individuals in the data), $d$ denotes the number of columns in the matrix, and we assume each entry is at most $1$. As before, $\epsilon$ denotes the privacy parameter and $k$ denotes the number of queries. Observe that we (again) compare $k$ \emph{predetermined} queries for non-interactive mechanisms with $k$ \emph{adaptively chosen} queries for interactive ones. The remaining parameters are omitted from this comparison. Results are summarized in Table~\ref{tab:mechanisms_comparison_var}.

\subsubsection{Our Algorithm}
\label{subsubsec:our_alg_var}

Our algorithm's utility is computed simply by plugging in $\nu = O(1/k)$ to Theorem~\ref{thm:utility_directional_variance}, which gives a utility bound of $O(\log(k)/\epsilon^2)$. 

\subsubsection{Na\"ively Adding Laplace Noise}
\label{subsubsec:lapalacian_var}

Again, the simplest alternative is to answer each directional-variance query with  $\Phi_x(A) + Lap(0,\epsilon^{-1})$. The composition theorem of~\cite{DRV10} assures us that for $k$ queries we preserve $(O(\sqrt{k} \epsilon),\delta)$-differential privacy.

\subsubsection{Randomized Response}
\label{subsubsec:RR_var}

We now consider a Randomized Response mechanism, similar to the Randomized Response mechanism of~\cite{GuptaRU12}. We wish to output a noisy version of $A\T A$, by adding some iid random noise to each entry of $A\T A$. Since we call two matrices neighbors if they differ only on a single row, denote $v$ as the difference vector on that row. It is simple to see that by adding $v$ to some row in $A$, each entry in $A\T A$ can change by at most $\|v\|_1$. Recall that we require $\|v\|_2=1$ and so $\|v\|_1 \leq \sqrt d$. Therefore, we have that in order to preserve $(\epsilon,\delta)$-differential privacy, it is enough to add a random Gaussian noise of $\mathcal{N}(0,\frac {d\log(d)} {\epsilon^2})$ to each of the $d^2$ entries of $A\T A$.


Next we give the utility guarantee of the Randomized Response scheme. Fix any unit length vector $x$. We think of the matrix we output as $A\T A + N$, where $N$ is a matrix of iid samples from $\mathcal{N}(0,\frac {d\log(d)} {\epsilon^2})$. Therefore, in direction $x$, we add to the true answer a random noise distributed like $x\T N x \sim \mathcal{N}(0,\left(\sum_{i,j}x_i^2x_j^2\right)\frac {d\log(d)} {\epsilon^2}) = \mathcal{N}(0,\frac {d\log(d)} {\epsilon^2})$. So w.h.p the noise we add is within factor of $\tilde O(\sqrt{d}/\epsilon)$ for each query, and for $k$ queries it is within factor of $\tilde O(\sqrt{d\log(k)}/\epsilon)$.

\subsubsection{The Multiplicative Weights Mechanism}
\label{subsubsec:multi_weights_var}

It is not straight-forward to adapt the Multiplicative Weights mechanism to answer directional variance queries. 
We represent $A\T A$ as a histogram over its $d^2$ entries (so the size of the ``universe'' is $N=d^2$), but it is not simple to estimate what is the equivalent of number of individuals in this representation. We chose to take the pessimistic bound of $nd^2$, since this is the $L_1$ bound on the sum of entries in $A\T A$, but we comment this is a highly pessimistic bound. It is fairly likely that the number of individuals in this representation can be set to only $O(d^2)$.

Plugging these parameters into the utility bounds of the Multiplicative Weights mechanism, we get a utility bound of $\tilde O(d\sqrt{n}\log(k)/\epsilon)$. Plugging them into the improved bounds of the IDC mechanism, we get $\tilde O(d\sqrt{n\log(k)/\epsilon})$. Observe that even if replace the pessimistic bound of $nd^2$ with just $d^2$, these bounds depend on $d$.

\begin{table}[t]
\footnotesize
\centering
\begin{tabular}{| p{2.4cm} |c|p{1.0cm}|p{1.0cm}|p{1.0cm}|p{2.5cm}|}
\hline
\centering Method & Additive Error & {Multi-plicative Error?} & {Inter-active?} & {Tract-able?} \cr
\hline
Laplace Noise \cite{Dwork06calibratingnoise}& $O(\sqrt{k}/\epsilon)$ & \centering\XSolid & \centering\Checkmark & \centering\Checkmark \cr
\hline
Randomized Response & $\tilde O(\sqrt{d\log(k)}/\epsilon)$ & \centering\XSolid & \centering\XSolid & \centering\Checkmark \cr 
\hline
\parbox{2cm}{MW~\cite{hardt2010multiplicative} \\ IDC~\cite{GuptaRU12}} &  \parbox{3.0cm}{\centering $\tilde O(d\sqrt{n}\log(k)/\epsilon)$ \\ $\tilde O(d\sqrt{n\log(k)/ \epsilon})$ } & 
\centering{\XSolid} & \centering{\Checkmark} &\centering{\Checkmark} \cr 
\hline
\hline
JL & $O(\log(k) /\epsilon^2)$ & \centering\Checkmark & \centering\XSolid & \centering\Checkmark \cr
\hline
\end{tabular}
\caption{Comparison between mechanisms for answering directional variance queries.\label{tab:mechanisms_comparison_var}}
\end{table}